\RequirePackage{fix-cm}
\documentclass[envcountsect,numbook]{svjour3_arxiv3}   
\smartqed  

\usepackage{graphicx}
\usepackage[utf8]{inputenc}
\usepackage{dsfont}
\usepackage{bbm}
\usepackage{subfigure}
\usepackage{color}
\usepackage{url}
\usepackage{enumerate} 
\usepackage{booktabs}
\usepackage{tabularx}
\usepackage{ltablex}
\usepackage{mathptmx}
\usepackage{amssymb,amsmath,amsfonts,latexsym} 
\numberwithin{equation}{section}
\usepackage{longtable} 
\usepackage{textcomp}
\usepackage{lscape}

\usepackage{graphicx}
\usepackage[utf8]{inputenc}
\usepackage{dsfont}
\usepackage{subfigure}
\usepackage{color}
\usepackage{url}
\usepackage{enumerate} 
\usepackage{booktabs}
\usepackage{tabularx}
\usepackage{ltablex}
\usepackage{mathptmx}
\usepackage{amssymb,amsmath,amsfonts,latexsym} 
\usepackage{longtable} 
\usepackage{textcomp}
\usepackage{lscape}

\usepackage[12pt]{extsizes}
\usepackage{geometry}
\usepackage{ifthen}
\usepackage{xcolor}
\usepackage{mathtools}
\mathtoolsset{showonlyrefs}
\usepackage{soul}
\DeclareMathAlphabet{\pazocal}{OMS}{zplm}{m}{n}
\let\mathcal\pazocal
\allowdisplaybreaks

\spnewtheorem{theorem}{Theorem}{\bfseries}{\itshape}
\spnewtheorem{cor}[theorem]{Corollary}{\bfseries}{\itshape}
\spnewtheorem{lemma}[theorem]{Lemma}{\bfseries}{\itshape}
\spnewtheorem{proposition}[theorem]{Proposition}{\bfseries}{\itshape}
\spnewtheorem{definition}[theorem]{Definition}{\bfseries}{\itshape}
\spnewtheorem{remark}[theorem]{Remark}{\bfseries}{\upshape}
\spnewtheorem{assumption}[theorem]{Assumption}{\bfseries}{\itshape}
\counterwithin{theorem}{section}
\smartqed

\definecolor{myred}{rgb}{0.8,0,0}  
\definecolor{orange}{rgb}{1,0.53,0}
\definecolor{mygreen}{rgb}{0,0.5,0}
\definecolor{myviolet}{rgb}{0.58,0,0.83}

\newcommand{\mymarginpar}[1]{ \marginpar{{\tiny #1}}}

\renewcommand{\paragraph}[1]{{\bf #1.}}
\renewcommand{\mymarginpar}[1]{} 
{\vskip\baselineskip\noindent\textbf{Proof of {#1}:}}%
{\hspace*{.1pt}\hspace*{\fill}\BOX\vskip\baselineskip}

\newcommand{\KBY}[1]{C_{Y,#1}}
\newcommand{\KBYreg}[1]{\widetilde{C}_{Y,#1}}

\renewcommand{\paragraph}[1]{{\bf #1.}}
\definecolor{myred}{rgb}{0.8,0,0}  
\definecolor{mygreen}{rgb}{0,0.7,0}
\definecolor{mygrey}{rgb}{0.5,0.5,0.5}

\def \R{\mathbb{R}}               
\def \N{\mathbb{N}}               
\def \P{\mathbb{P}}             
\def \1{{\bf 1}}                
\def \0{{\bf 0}}
\def\utility{\mathcal U}

\def\qed{\hfill$\Box$}

\def\eps{\varepsilon}

\def\revlevel{\overline{\mu}}
\def\revspeed{\kappa}
\def\driftinitial{\overline{\filter}_0}
\def\filterinitial{\filter_0}
\def\covinitial{\overline{\variance}_0}
\def\condcovinitial{\variance_0}

\def\restri{\mathcal R}

\def\nAktien{d}
\def\YQ{G}
\def\nQ{\nAktien_{\YQ}}
\def\zsmall{g}
\def\nZustand{{d_Y}}

\def\alphaYQ{{\alpha}_{\YQ}}
\def\alphaYQquer{\underline{\alpha}_{\YQ}}
\def\gammaYQ{\gamma_{\YQ}}

\def\variance{q}
\def\pointp{p}

\def\HC{Z}

\def\HF{F}
\def\HR{R}

\def\nWienerRendite{d_1}
\def\nWienerDrift{d_2}

\def\varianceexp{\Gamma}

\def\volR{\sigma_R}
\def\voldrift{\sigma_{\mu}}

\def\filter{m}
\def\alphamm{\alpha_{\Mpro}}
\def\betam{\beta_{\Mpro}}
\def\gammam{\gamma_{\Mpro}}

\def\alphaqq{\alpha_{\Qpro}}
\def\betaq{\beta_{\Qpro}}
\def\gammaq{\gamma_{\Qpro}}
\def\drift{\mu}

\def\komppoi{\widetilde{N}^{\lambda}}

\def\Radon{\Lambda}

\def\alphay{\alpha_Y}
\def\betay{\beta_Y}
\def\gammay{\gamma_Y}
\def\generator{\mathcal L}
\def\reward{D}
\def\valuefkt{V}
\def\valueorigin{\mathcal V}

\def\rewardorigin{\mathcal D}

\def\stateprocess{Y}
\def\wealth{X}

\def\Jpi{\overline \rewardorigin_0}

\def\Mpro{M}
\def\Qpro{Q}
\def\etaT{\eta}

\def\np{k}
\def \dist{\text{\rm dist}\,}

\newcommand\norm[1]{\left\lVert#1\right\rVert}

\def \trace{\operatorname{tr}}
\def \E{{\mathbb{E}}}

\def\Statespace{{\mathcal S}}

\def\diffop{D}
\def\AlgT{\mathcal{F}^T}
\def \one{\mathbbm{1}}

\begin{document}
	
	\title{Portfolio Optimization in a Market with Hidden Gaussian Drift and
		Randomly Arriving Expert Opinions: Modeling and Theoretical Results}

	\titlerunning{Portfolio Optimization with Randomly Arriving Expert Opinions}
	\author{Abdelali Gabih \and Ralf Wunderlich}
	
	\authorrunning{A.~Gabih,  R.~Wunderlich} 
	
	\institute{Abdelali Gabih, corresponding author   \at
		Equipe de Modélisation et Contrôle des Systèmes Stochastiques et Déterministes (MC2SD), Faculty of Sciences, Chouaib Doukkali
		University Morocco, El Jadida 24000,   \email{\texttt{a.gabih@uca.ma}}
		\and 
		Ralf Wunderlich\at
		Brandenburg University of Technology Cottbus-Senftenberg, Institute of Mathematics, P.O. Box 101344, 03013 Cottbus, Germany;  
		\email{\texttt{ralf.wunderlich@b-tu.de}}\newline \newline
		The authors thank   Hakam Kondakji (Helmut Schmidt University Hamburg), Dorothee Westphal and Jörn Sass  (TU Kaiserslautern) for valuable discussions that improved this paper.
	}

	\date{}
	
	\maketitle

	\abstract{
		This paper investigates the optimal selection of portfolios for power utility maximizing investors in a financial market where stock returns depend on a hidden Gaussian mean reverting drift process. Information on the drift is obtained from returns and expert opinions in the form of noisy signals about the current state of the drift arriving randomly over time. The arrival dates are modeled as the jump times of a homogeneous Poisson process. Applying  Kalman filter techniques we derive estimates of the hidden drift which are described by the conditional mean and covariance of the drift given the observations. The utility maximization problem is solved with dynamic programming methods. We derive the associated dynamic programming equation and study regularization arguments for a rigorous mathematical justification.}
	
	

	\section{Introduction}
	The selection of optimal portfolios  is  one of the  classical problems in mathematical finance.  The
	main objective is to find the best allocation  of the investors capital between the available assets to maximize some performance criterion. Constructing  an optimal portfolio is challenging since the market participants  operate
	in an uncertain environment.  		
	This fact was already addressed  by   Merton \cite{Merton (1971)}, Brennan, Schwartz,
	and Lagnado \cite{Brennan et al (1997)}, and Bielecki and Pliska \cite{Bielecki_Pilska (1999)}. They  argue that stochastic factors explain observed variations in expected asset returns, and accordingly optimal trading strategies depend crucially on the drift of the underlying asset return processes.
	However, practitioners face a challenging task when trying to get accurate empirical estimations of the drift based exclusively on historical asset prices.  Rogers \cite[Chapter 4.2]{Rogers (2013)} addressed in his celebrated '20s example'  that even for a constant drift  the volatility in a typical financial asset overwhelms the drift to such an extent that one needs centuries of data  to form reliable estimates of the drift.  Such data is of course not available and it is also quite unrealistic to consider a constant drift over longer periods of time.  Instead, drift models should allow for random and non-predictable fluctuations and drift effects overshadowed by volatility.
	
	In order to tackle the problem of constructing good drift estimates which are needed for the construction of trading decisions, practitioners  do not only rely on historical asset returns but  also  on external sources of information such as news, company reports, ratings or their own intuitive views on the future asset performance.  In the literature these outside sources of information are coined \textit{expert opinions}. They are used to improve drift estimates. A mathematical formalization			of this procedure for a static one-period financial market model is known in the literature as Black-Litterman model  which can be considered as an	extension of the classical one-period Markowitz model, see Black and Litterman~\cite{Black_Litterman (1992)}. The approach there consists in incorporating expert opinions to  improve return predictions  by means of a  Bayesian updating of the drift estimate. 
	
	Contrary to this classic approach in this paper we consider a continuous-time financial market model and allow expert opinions to arrive repeatedly over time. In   Davis and Lleo  \cite{Davis and Lleo (2013_1)} this approach is coined  ``Black--Litterman in Continuous Time'' (BLCT). We work with a  hidden Gaussian model where asset prices follow a diffusion process whose drifts  are driven by an unobservable mean-reverting Gaussian process.			
	Information on the drift is available from two sources. First, investors observe continuously stock prices or equivalently the return process.  Second,  investors exploit expert opinions delivered  at random time points in a form of unbiased drift estimates. The arrival dates are described by the jump times of a homogeneous Poisson process with known intensity. This setting is well-suited to model the arrival of information such as news, ratings.  Further,  also so-called 'alternative data'  from sources outside the traditional data about companies and financial markets such as social media posts or product review trends are characterized by a non-regular and non-predictable arrival of  information. The setting with random arrival dates complements the approach presented in Gabih et al.~\cite{Gabih et al (2023) PowerFix} in which the authors assume fixed and known arrival dates. The latter is a suitable model for regularly arriving company reports and the views of appointed analysts.

	For portfolio selection problems the task is to reflect  the informational advantage of expert opinions in the dynamic trading strategies. 
	It is well-known that  the construction of  good trading strategies depends on the quality of the hidden drift estimation. One way to specify such an estimation is combining  filtering results for hidden Gaussian models and Bayesian updating to compute the conditional distribution of the drift given the available information drawn from  the return observations and expert opinions. In the case of investors having access only to the return process that filter is known as the classical  Kalman--Bucy filter, see for example Liptser and Shiryaev~\cite{Liptser-Shiryaev}. Based on this,  the filter for investors combining both returns and expert opinions can be derived  by a Bayesian update of the  current drift estimate at each information date.  This idea leads to the above mentioned   continuous-time version of the static Black-Litterman approach.

	The literature on  utility maximization problems for partially informed investors mainly uses  two  models for the hidden drift. The first group works with asset price models in which the drift is a  Gaussian mean reverting process such   by Lakner \cite{Lakner (1998)} and  Brendle \cite{Brendle2006}. In a second group  the drift is described by a continuous-time hidden Markov chain  as in   Rieder and Bäuerle \cite{Rieder_Baeuerle2005}, Sass and Haussmann \cite{Sass and Haussmann (2004)}. For a  good overview with further references and generalizations  we refer to  Björk et al.~\cite{Bjoerk et al (2010)}.

	The literature on BLCT in which expert opinions are included comprises  a series of papers \cite{Gabih et al (2014),Gabih et al (2019) FullInfo,Sass et al (2017),Sass et al (2021),Sass et al (2022)} 
	of the present authors and of Sass and Westphal. The case of power utility maximization for expert opinions arriving at fixed and known time points is studied in  Gabih et al. \cite{Gabih et al (2023) PowerFix}. Similar problems for  continuous-time expert opinions are treated in a series of papers by Davis and Lleo, see \cite{Davis and Lleo (2013_1),Davis and Lleo (2020)} and the references therein. Power utility maximization problems for drift processes described by  continuous-time hidden Markov chains have been studied in  Frey et al.~\cite{Frey et al. (2012),Frey-Wunderlich-2014}. Finally,  the well posedness of the power utility maximization problem is  addressed in Gabih et al. \cite{Gabih et al (2022) Nirvana}. 
	
	\smallskip
	In the present paper we study in detail the case of power utility maximization for market models with a hidden Gaussian drift and randomly arriving discrete-time expert opinions. In the literature this case  was only considered in the PhD thesis of Kondakji \cite{Kondkaji (2019)} which serves in this paper as starting point for various extensions and generalizations. The utility maximization problem is treated as a stochastic optimal control problem and solved using dynamic programming methods. 	We  apply a change of measure technique which  allows to study  simplified control problems with  a risk-sensitive performance criterion in which the state variables are reduced  to the Kalman filter processes of conditional mean and covariance.

	The setting  with expert opinions arriving at already known times  is explored in  Gabih et al.~\cite{Gabih et al (2023) PowerFix}. In this work  the  conditional covariance is deterministic and can be computed offline and therefore removed from the state variables. Then, the state process between two arrival dates is a pure diffusion. This allows for a closed-form solution of the associated dynamic programming equation in terms of a backward recursion for the value function as well as for the optimal strategy. Contrary to this, in the current setting with randomly arriving expert views the conditional covariance is a stochastic process with jumps at the random  arrival dates and  deterministic with continuous paths between those dates. Therefore, the conditional covariance  has to be included in the state variables. The dynamic programming equation for the resulting control problem is a partial integro-differential		equation (PIDE) which is degenerated in the diffusion part of the differential operator. This precludes
	the application of classical existence and uniqueness results for the solution of the PIDE
	as well as corresponding verification theorems for the control problem.
	\smallskip
	
	\paragraph{Our contribution}
	We tackle the above mentioned problems by using a regularization technique motivated by similar approaches in Frey et al. \cite{Gabih et al (2014)} and Shardin and Wunderlich \cite{Shardin and Wunderlich (2017)}. Note that in these papers the  drift is a finite-state	Markov chain and hence bounded. However,  our setting  involves an unbounded
	Ornstein-Uhlenbeck drift process which creates a number of new and delicate technical difficulties. Starting point of this regularization  is  a ``small'' perturbation of the dynamics of the	state process of the control problem by another driving  Brownian motion of appropriate dimension and independent of the already existing drivers.  This perturbation is  scaled by a small parameter which later is sent to zero. Then the dynamic programming equation associated with the regularized optimization problem is strictly elliptical so that classical verification results can be applied. For this family of perturbed control problems, we first show that for scaling parameters approaching zero, the solutions of the regularized state equations  converge in the mean-square sense to the solution of the original, non-regularized state equation.	 Next,  we show that the reward functions of the control problems converge uniformly for all admissible strategies. For the present risk-sensitive performance criterion this appears  as a problem of convergence of exponential moments for which  we apply the concept of uniform integrability.  Finally, we  prove that   value functions of the regularized problem converge to the value function of the original problem. That result implies that optimal strategies for the regularized problem  are nearly or	$\varepsilon$-optimal in the original problem which gives a method for computing $\varepsilon$-optimal strategies. We finally point out that the theoretical findings of this paper are illustrated in ~\cite{Kondkaji (2019)} by results of extensive numerical experiments based on the numerical solution of the PIDE appearing in the dynamic programming equation.
	\smallskip
	
	\paragraph{Outline}
	The paper is organized as follows. In Sec.~\ref{market_model} we
	introduce our continuous-time model for the financial market including the discrete-time expert
	opinions delivered randomly over time. Further,  we formulate the portfolio optimization problem under partial information. Sec.~\ref{Filtering} is devoted to the estimation of the hidden drift and states the generalized Kalman filter equations for the corresponding conditional mean and conditional covariance process. In Sec.~\ref{Utility_Max}
	the power utility maximization problem   is reformulated as an equivalent stochastic optimal control problem with a risk-sensitive performance criterion. The associated dynamic programming equation to that problem is derived in Sec.~\ref{UtilityMaxDiscreteExperts}.  Finally, we present in Sec.~\ref{regularization} our main results concerning the mathematical justification of the dynamic programming equation using the regularization approach which allows  to compute $\varepsilon$-optimal strategies.  An appendix collects technical proofs which were removed from the main text.

	\medskip
	\paragraph{Notation} Throughout this paper, we use the notation $I_d$ for the identity matrix in $\R^{d\times d}$,  $0_{d}$
	denotes the null vector in $\R^d$ and  $0_{d\times m}$ the null matrix in $\R^{d\times m}$. 
	For a symmetric and positive-semidefinite matrix $A\in\R^{d\times d}$ we call a symmetric and positive-semidefinite matrix $B\in\R^{d\times d}$ the \emph{square root} of $A$ if $B^2=A$. The square root is unique and will be denoted by $A^{1/2}$. For a vector $X$ we denote by  $\norm{X}$ the maximum norm, it induces  for a matrix $A$ the row sum norm denoted by $\lVert A\rVert$.

	\section{Financial Market and Optimization Problem}
	\label{market_model}
	\subsection{\it\textbf{ Price Dynamics}}
	\label{PriceDynamics}  We model a financial market with one
	risk-free and multiple risky assets. The  setting is based on Gabih
	et al.~\cite{Gabih et al (2019) FullInfo,Gabih et al (2023) PowerFix} and
	Sass et al.~\cite{Sass et al (2017),Sass et al (2022),Sass et
		al (2021)}.
	For a  fixed date $T>0$ representing the investment horizon, we work
	on a filtered probability space $(\Omega,\mathcal{G},\mathbb{G},\P)$
	with filtration $\mathbb{G}=(\mathcal {G}_t)_{t \in [0,T]}$
	satisfying the usual conditions. All processes are assumed to be
	$\mathbb{G}$-adapted.
	We consider a market model  for one risk-free asset 
	and $\nAktien$ risky securities. We fix the risk-free asset as numéraire so that the  risk-free asset has a normalized	constant price $S^0_t=1$.
	The excess  log returns or risk premiums $R=(R^{1},\ldots,R^{\nAktien})$ of the risky securities  are described by stochastic processes defined by the SDE
	\begin{align}
		dR_t=\mu_t\, dt+\volR\, dW^{R}_t, \label{ReturnPro}
	\end{align}
	driven by  a $\nWienerRendite$-dimensional $\mathbb{G}$-adapted	Brownian motion $W^{\HR}$ with  $\nWienerRendite\geq\nAktien$. In the remainder of this paper we will call $R$ simply \textit{returns}.    $\mu=(\mu_t)_{t\in[0,T]}$ denotes the stochastic drift process which is described in detail below.	 The volatility matrix $\volR\in\mathbb
	R^{\nAktien\times\nWienerRendite}$ is  assumed to be constant over
	time such that $\Sigma_{R}:=\volR\volR^{\top}$ is positive definite. Accordingly, the discounted price process $S=(S^1,\ldots,S^{\nAktien})$ of
	the risky securities reads as
	\begin{align}
		dS_t&=diag(S_t)\, dR_t,~~ S_0=s_0, \label{stockmodel}
	\end{align}
	with some fixed initial value $s_0=(s_0^1,\ldots,s_0^d)$. Note
	that for the solution to the above SDE it holds
	\begin{align*}
		\log S_t^{i}-\log s_0^{i} &= \int\nolimits_0^t \drift_s^{i}ds
		+\sum\limits_{j=1}^{\nWienerRendite}\Big(
		\sigma_R^{ij}W_t^{R,j}-\frac{1}{2} (\sigma_R^{ij})^2 t\Big)
		=R_t^{i}-\frac{1}{2}\sum\limits_{j=1}^{\nWienerRendite}
		(\sigma_R^{ij})^2 t ,\quad i=1,\ldots,\nAktien.
	\end{align*}
	So we have the equality $\mathbb{G}^R = \mathbb{G}^{\log S} =
	\mathbb{G}^S$, where  for a generic process $X$ we denote by
	$\mathbb{G}^X$ the filtration generated by $X$. This is useful since
	it allows to work with $R$ instead of $S$ in the filtering part.
	
	The dynamics of the drift process $\mu=(\mu_t)_{t\in[0,T]}$ in \eqref{ReturnPro}
	are
	given by the stochastic differential equation (SDE)
	\begin{eqnarray}
		\label{drift} d\mu_t=\revspeed(\revlevel-\mu_t)\, dt+\voldrift\,
		dW^{\mu}_t,
	\end{eqnarray}
	defining an $d$-dimensional mean reverting process known as  Ornstein-Uhlenbeck process. Here $\revspeed\in\mathbb R^{\nAktien\times\nAktien}$,
	$\voldrift\in\mathbb R^{\nAktien\times\nWienerDrift}$ and
	$\revlevel\in\mathbb R^{\nAktien} $ are constants such that all eigenvalues of $\revspeed$ have
	a positive real part (that is, $-\revspeed$ is a stable matrix), $\Sigma_{\mu}:=\voldrift\voldrift^{\top}$
	is positive definite, and $W^{\mu}$ is a
	$\nWienerDrift$-dimensional Brownian motion such that
	$\nWienerDrift\geq\nAktien$.    The parameters $\revlevel, \revspeed, \voldrift$ are called mean-reversion level,
	the mean-reversion speed and volatility of $\mu$, respectively. The initial value $\drift_0$ is assumed to be a
	normally distributed random variable independent of $W^{\mu}$ and
	$W^{R}$ with mean $\driftinitial\in \R^{\nAktien}$ and covariance
	matrix $\covinitial\in\mathbb R^{\nAktien\times\nAktien}$ assumed to
	be symmetric and  positive semi-definite.   For the sake of simplification and
	shorter notation we assume  that the Wiener processes $W^{R}$ and
	$W^{\mu}$ driving the return and drift process, respectively, are
	independent. We refer to Brendle \cite{Brendle2006},  Colaneri
	et al. \cite{Colaneri et al (2021)} and  Fouque et al.~\cite{Fouque et al. (2015)} for
	the general case.
	
	\subsection{\it\textbf{ Expert Opinions }}
	\label{Expert_Opinions}  
	We assume that investors do not have a complete access to the market information, they can instead observe the historical data of the return process  $R$ but they neither observe the factor process $\mu$ nor the Brownian motion $W^{R}$. Further, investors are aware of  
	the model
	parameters such as $\volR,\revspeed, \revlevel, \voldrift $  and
	the distribution $\mathcal{N}(\driftinitial,\covinitial)$ of  the
	initial value $\drift_0$. Accordingly, information about the drift $\mu$ can be
	drawn from observing the returns $R$. A special feature of our model
	is that practitioners may also rely on  additional sources of information
	about the drift in form of \textit{expert opinions} such as news,
	company reports, ratings or their own intuitive views on the future
	asset performance. These expert opinions are delivered at a random points in time
	$T_k$, and  provide  noisy signals about
	the current state of the drift. We model these expert opinions  by a marked point process
	$(T_k,Z_k)_k$, so that at $T_k$ the investor observes the
	realization of a random vector $Z_k$ whose distribution depends on
	the current state $\mu_{T_k}$ of the drift process. The arrival
	dates $T_k$ are modelled as jump times of a standard Poisson process
	with intensity $\lambda>0$, independent of $\mu$, so that the timing
	of the information arrival does not carry any useful information
	about the drift.  A generalization to the case of a time-dependent intensity expressed by some deterministic function $\lambda(t) >0$ with $ \int_0^T\lambda(s)ds<\infty $ is straightforward. For ease of exposition we restrict to the constant intensity case.
	Further, for the sake of convenience we also write $T_0:=0$ although no
	expert opinion arrives at time $t=0$.
	
	The signals or ``the expert's views'' at time $T_k$ are modelled by
	$\R^\nAktien$-valued  Gaussian random vectors
	$Z_k=(Z_k^1,\cdots,Z_k^{\nAktien})^{\top}$ with
	\begin{align}
		\label{Expertenmeinungen_fest}
		Z_k=\drift_{T_k}+{\varianceexp}^{\frac{1}{2}}\mathcal{E}_k,\quad
		k=1,2,\ldots,
	\end{align}
	where the matrix  $\varianceexp\in\R^{\nAktien\times\nAktien}$ is
	symmetric and positive definite,  
	$\mathcal{E}_k\sim \mathcal{N}(0,I_d), k\ge 1$ 
	is a sequence of independent
	standard normally distributed random vectors taken to be independent of
	both the Brownian motions $W^R, W^\mu$ and of the initial value $\mu_0$
	of the drift process. Accordingly, that means that, given $\mu_{T_k}$, the expert
	opinion $Z_k$ is $\mathcal{N}(\mu_{T_k},\varianceexp)$-distributed so that
	$Z_k$ can be considered as an unbiased estimate of the unknown
	state of the drift at time $T_k$. The matrix $\varianceexp$ is a
	measure of the expert's reliability. 
	The diagonal entry $\varianceexp_{ii}, i=1,\ldots,d,$ is the variance of the $i$-th component of the expert's
	estimate of the drift at time $T_k$. The larger $\varianceexp_{ii}$ the less
	reliable is that expert's view.
	
	\smallskip
	The above model can be modified such that expert opinions arrive at fixed and known dates   instead of random dates.
	That approach together with results for filtering  and maximization of log-utility was studied in detail in Sass et al.~\cite{Sass et al (2021)}. For maximization of power utility in such a model we refer  Gabih et al.~\cite{Gabih et al (2023) PowerFix}.
	
	One may also allow for relative expert views where experts
	give an estimate for the difference in the drift of two stocks
	instead of absolute views. This extension  can be studied in
	Sch\"ottle et al.~\cite{Schoettle et al. (2010)} where the authors
	show how to switch between these two models for expert opinions by
	means of a pick matrix containing the information which assets are affected by the respective forecasts.

	\subsection{\it \textbf{ Investor Filtration}}
	\label{Investor_Filtration}   
	The information available to
	the investor combining return observations with the discrete-time expert  opinions
	$Z_k$ is described by the \textit{investor filtration}		$\mathbb{F}^{\HC}=(\mathcal{F}^{\HC}_t)_{t\in[0,T]}$ where for $t\in [0,T]$ the  $\sigma$-algebra $\mathcal{F}^{\HC}_t$ is defined
	\[\begin{array}{rl}
		\mathcal {F}_t^{\HC}&=\sigma(R_s, s\le t,\,~(T_k,Z_k),~ T_k\le t) \vee \mathcal{F}_0^I, \\[0.5ex]
	\end{array}
	\]
	where we  assume that the above $\sigma$-algebras $\mathcal{F}_t^{\HC}$	are augmented by the null sets of $\P$. 
	The $\sigma$-algebra $\mathcal{F}_0^I$ models  prior knowledge about the drift process at time $t=0$, e.g., from
	observing  returns  or expert opinions in the past, before the
	trading period $[0,T]$. 	The conditional distribution	of the initial value drift $\mu_0$ given $\mathcal{F}_0^{ I}$ is assumed to be the
	normal distribution $\mathcal{N}(m_0,q_0)$ with mean	$\filterinitial\in \R^{\nAktien}$ and covariance matrix given by the symmetric and  positive semi-definite matrix 
	$\condcovinitial\in\mathbb R^{\nAktien\times\nAktien}$. An  investor with information described by the filtration 	$\mathbb{F}^{\HC}$ will be  called $\HC$-investor.
	For investor filtrations describing alternative information regimes such as investors observing only returns,  combining return observations with the continuous-time expert  opinions, or fully informed investors observing the drift directly we refer  Gabih et al.~\cite{Gabih et al (2023) PowerFix}.

	\subsection{\it \textbf{Portfolio and Optimization Problem}}
	We describe the self-financing trading of an investor by the initial
	capital $x_0>0$ and the  $\mathbb{F}^{\HC}$-adapted trading strategy
	$\pi=(\pi_t)_{t\in[0,T]} $ where $\pi_t\in\R^{\nAktien}$. Here
	$\pi_t^{i}$ represents the proportion of wealth invested in the
	$i$-th stock at time $t$.  The assumption that $\pi$ is
	$\mathbb{F}^{\HC}$-adapted models that investment decisions have to be
	based on the information available to the $\HC$-investor which he obtains
	from observing assets prices combined with expert
	opinions. Following
	the strategy $\pi$   the investor generates a wealth process
	$(X_t^{\pi})_{t\in [0,T]}$ whose dynamics  reads as
	\begin{eqnarray} \label{wealth_phys}
		\frac{dX_t^{\pi}}{X_t^{\pi}}= \pi_t^{\top}dR_t &= &
		\pi_t^{\top}\mu_t\; dt+\pi_t^{\top}\volR\; dW_t^{R},\quad
		X_0^{\pi}=x_0.
	\end{eqnarray}
	We denote by 
	\begin{align}
		\begin{aligned}
			\label{set_admiss_0} \mathcal{A}_0 =\Big\{&\pi= (\pi_t)_{t}  \colon
			\pi_t\in\mathbb R^{\nAktien}, \text{ $\pi$ is $\mathbb{F}^{\HC}$-adapted
			}, X^\pi_t > 0, \E\Big[ \int\nolimits_0^T \|\pi_t\|^2\, dt
			\Big]<\infty, \\
			&  \text{the process } \Radon \text{ defined below in \eqref{Radon} satisfies } \E\big[\Radon_T\big]=1  \Big\}
		\end{aligned}
	\end{align}
	the  class of {\em admissible trading strategies}.  The last condition in the definition of $\mathcal{A}_0$ is  needed to apply  a change of measure technique for the solution of the optimization problem. More  details  are provided below in Sec.~\ref{Utility_Max}.
	
	We assume that
	the investor wants to maximize the expected utility of terminal
	wealth for a given utility function $U : \R_+\rightarrow\R$
	modelling the risk aversion of the investor. In our approach we will
	use  the  power utility function
	\begin{align}
		\label{util_def}
		\utility_{\theta}(x):=\frac{x^{\theta}}{\theta},\quad
		\theta\in(-\infty,0)\cup(0,1).
	\end{align}		
	The optimization problem thus reads as 
	\begin{align}
		\label{opti_org} 
		\mathcal V_0(x_0):=\sup\limits_{\pi\in\mathcal{A}_0}
		\rewardorigin_0(x_0;\pi) \quad \text{where}\quad
		\rewardorigin_0(x_0;\pi) =
		\E\left[\utility_{\theta}(X_T^{\pi})~|~\mathcal{F}^{\HC}_0\right],~\pi\in\mathcal
		A_0,
	\end{align}
	where we call $\rewardorigin_0(x_0;\pi)$ \textit{reward function}
	or \textit{performance criterion} of the strategy $\pi$ and $
	\mathcal V_0(x_0)$ \textit{value function} to given initial
	capital $x_0$.
	This is a maximization problem under partial
	information since we have required that the strategy $\pi$ is
	adapted to the investor filtration $\mathbb F^{\HC}$ while the drift
	coefficient of the wealth equation \eqref{wealth_phys} is not
	$\mathbb F^{\HC}$-adapted, it depends on the non-observable drift
	$\mu$. Note that for $x_0 > 0$ the solution to SDE
	\eqref{wealth_phys} is strictly positive. This guarantees that
	$X_T^{\pi}$ is in the domain of the power utility function $\utility_\theta$.
	
	\subsection{\it \textbf{Well Posedness of the Optimization Problem}}
	\label{Subsec:WellPosedness}		
	A rigorous analysis of the   utility maximization problem
	\eqref{opti_org} requires that it is well-posed in the sense that  $\valueorigin_0^{H}(x_0)$ is finite. This ensures that the maximum expected
	utility  of terminal wealth cannot explode in finite time. The latter is	the case for  so-called nirvana strategies described e.g. in Kim and
	Omberg \cite{Kim and Omberg (1996)} and Angoshtari	\cite{Angoshtari2013}. Nirvana strategies generate in finite time  a
	terminal wealth with a distribution leading to infinite expected 	utility although the realizations of terminal wealth  may be finite.	
	In Gabih et al.~\cite{Gabih et al (2023) PowerFix} the authors  show that for power utility with $\theta < 0$ problem \eqref{opti_org} is always well-posed. For parameters $\theta\in(0,1)$ we give conditions to the model parameters ensuring well posedness. In the following we always assume that these conditions are satisfied.

	\section{Partial Information and Filtering}
	\label{Filtering}  
	The trading decisions of investors are based on
	their knowledge about the drift process $\mu$. While the
	fully informed investor observes the drift directly, the $\HC$-investor has to estimate it. This leads us to a filtering problem
	with hidden signal process $\mu$ and observations given by the
	returns $R$ and  the sequence of expert opinions $(T_k,Z_k)$. The \textit{filter} for
	the drift  $\mu_t$ is its projection on the
	$\mathcal{F}_t^{\HC}$-measurable random variables described by the
	conditional distribution of the drift given $\mathcal{F}_t^{\HC}$. The
	mean-square optimal estimator for the drift at time $t$, given the
	available information is  the \textit{conditional mean}
	$$\Mpro_t:=E[\mu_t|\mathcal{F}_t^{\HC}].$$
	The accuracy of that estimator  can be described by the
	\textit{conditional covariance matrix}
	\begin{align}\nonumber
		\Qpro_t:=E[(\mu_t-\Mpro_t)(\mu_t-\Mpro_t)^{\top}|\mathcal{F}^{\HC}_t].
	\end{align}
	Since in our filtering problem  the signal   $\mu$, the observations
	and the initial value of the filter  are jointly Gaussian also the
	filter distribution is Gaussian and completely characterized by the
	conditional mean $\Mpro_t$ and the conditional covariance
	$\Qpro_t$. 
	
	\medskip
	\paragraph{Dynamics of the filter}
	We now  give the dynamics of the filter  processes $\Mpro$ and 	$\Qpro$. The next lemma provides the filter for the $\HC$-investor who
	combines continuous-time observations of  stock
	returns and expert opinions received at discrete points in time. For a detailed proof we refer to  Lemma 2.3 in \cite{Sass et al
		(2017)} and  Lemma 2.3 in \cite{Sass et al (2022)}.
	\begin{lemma}
		\label{filter_C}For the $\HC$-investor the filter  is Gaussian and the
		conditional distribution of the drift $\mu_t$ given $\mathcal
		F_t^{\HC}$ is the normal  distribution $\mathcal
		N\left(\Mpro_t,\Qpro_t\right)$.
		\ \\[-2ex]
		\begin{enumerate}
			\item[(i)]
			Between two information dates $T_k$ and $T_{k+1}$, $k\in\N_0$,   the
			conditional mean  $\Mpro_t$  satisfies 
			\begin{align}
				\label{filter_C0}
				d\Mpro_t&=\revspeed(\revlevel-\Mpro_t)\;dt+\Qpro_t\,
				\Sigma_R^{-1/2}\;d\widetilde{W}_t \quad\text{for}~~ t\in
				[T_k,T_{k+1}).
			\end{align}
			The  innovation process
			$\widetilde{W}=(\widetilde{W}_t)_{t\in[0,T]}$   given by
			$$d\widetilde{W}_t=\Sigma_R^{-1/2}\left(dR_t-\Mpro_t
			dt\right), \quad \widetilde{W}_0=0,$$  is a standard  Brownian motion adapted to $\mathbb{F}$.\\
			The conditional covariance $\Qpro$ satisfies the ordinary Riccati
			differential equation 
			\begin{align*}
				d\Qpro_t&=(\Sigma_{\mu}-\revspeed
				\Qpro_t-\Qpro_t \revspeed^{\top}-\Qpro_t \Sigma_{R}^{-1}
				\Qpro_t)\; dt.
			\end{align*}
			The initial values are $\Mpro_{T_k} $ and $\Qpro_{T_k}$,
			respectively, with $\Mpro_{0}=\filterinitial$ and~
			$\Qpro_{0}=\condcovinitial$.
			\item[(ii)]
			At the information dates $T_k$,  $k\in\N$, the conditional mean  and
			covariance $\Mpro_{T_k}$ and $\Qpro_{T_k}$ are obtained from the
			corresponding values at time $T_{k^-}$ (before the arrival of the
			view) using the update formulas
			\begin{align}
				\label{filter_C_update}
				\Mpro_{T_k}=\rho_k\Mpro_{T_k-}+(I_d-\rho_k)Z_k
				\quad \text{and}\quad
				\Qpro_{T_k}=\rho_k \Qpro_{T_k-},
			\end{align}
			with the update factor
			$\rho_k=\varianceexp(\Qpro_{T_k-}+\varianceexp)^{-1}$.
		\end{enumerate}
	\end{lemma}
	Note that the values at an information date	$T_k$ are obtained from a Bayesian update. Further, 
	the dynamics of $\Mpro$ and $\Qpro$ between
	information dates are the same as for the  investor observing only returns but no expert opinions, see Gabih et al. \cite[Lemma 3.1]{Gabih et al (2023) PowerFix}. 
	In that setting the above conditional mean $\Mpro$ is
	a diffusion process and the conditional covariance $\Qpro$ is
	deterministic. Contrary to that,  for the  $\HC$-investor considered in this work  the conditional mean $\Mpro$ is a jump-diffusion process and the conditional covariance $\Qpro$ is no longer deterministic since the updates
	lead to jumps at the random arrival dates  $T_k$ of the expert
	opinions. Hence, $\Qpro$ is a piecewise deterministic stochastic
	process. This process enjoys the following boundedness property.
	
	\begin{proposition}[Sass et al. \cite{Sass et al (2022)}, Lemma 2.4]
		\label{properties_filter}\\  
		There exists a constant $C_{\Qpro}>0 $ such that $\norm{\Qpro_t} \le  C_{\Qpro} $ for all
		$t\in[0,T]$.			
	\end{proposition}

	\medskip
	\paragraph{Unified representation of filter dynamics}
	For an effecient treatment  of the utility maximization problems in Sec.~\ref{Utility_Max} it will be helpful to  express the dynamics of $\Mpro$ and $\Qpro$ given in Lemma \ref{filter_C} in a way that comprises both the behavior between information dates and the jumps at times $T_k$. 
	For this we need the following  result on the conditional distribution of the expert's opinion $Z_k$
	given the available information before the information date $T_k$. 
	\begin{lemma}[Kondakji \cite{Kondkaji (2019)}, Lemma 3.1.6]
		\label{bedingte_Verteilung_Z_lem} \\The conditional distribution of
		the expert opinions $Z_k$ given $\mathcal F_{T_k-}^\HC$ is the
		multivariate normal distribution  $\mathcal
		N\left(\Mpro_{T_{k-}},\varianceexp+\Qpro_{T_k-}\;\right)$,
		$k=1,2,\ldots$.
	\end{lemma}
	Next, we  work with a Poisson random measure as in Cont and Tankov~\cite[Sec.~2.6]{cont_tankov_2004}.
	Let $E=[0,T]\times\R^d$ and let $ \mathcal{E}_k$, $k=1,2,\dots$, be a sequence of independent multivariate standard Gaussian random variables on $\R^d$. For any  Borel set  $I \subset [0,T]$ and $B\subset \R^d$ let
	\[ N(I\times B)=\sum_{k\colon T_k\in I} \one_{\{ \mathcal{E}_k\in B\}} \]
	denote the number of jump times in $I$ for which  $ \mathcal{E}_k$ takes a value in $B$. Then  $N$ defines a Poisson random measure with a corresponding compensated measure $\widetilde{N}^\lambda(ds,du)=N(ds,du)-\lambda\,ds\,\varphi(u)\,du$, where $\varphi$ is the multivariate standard normal density on $\R^d$, see Cont and Tankov~\cite[Sec.~2.6.3]{cont_tankov_2004}.

	The next lemma rewrites
	the dynamics of  the filter processes  $\Mpro$ and $\Qpro$
	given in Lemma \ref{filter_C} and provides a semi-martingale
	representation  which is driven by the two martingales
	$\widetilde{W}$ and $\komppoi$.  
	For a detailed proof and further explanations we refer to Westphal~\cite[Prop.~8.14]{westphal_2019} and Kondakji \cite[Sec.~3.1]{Kondkaji (2019)}. 
	\begin{lemma}
		\label{Filter_C_Int} The dynamics of the conditional
		mean $\Mpro$ and the conditional covariance matrix $\Qpro$,
		respectively, are given by  
		\begin{align}
			\label{Filter_CN_3}
			d\Mpro_t&=\revspeed(\revlevel-\Mpro_t)\,dt
			+\betam(\Qpro_t)\,d\widetilde{W}_t
			+\int\nolimits_{\mathbb R^{\nAktien}} \gammam(\Qpro_{t-},u)\; \komppoi(dt,du),\\
			\label{Riccati_CN_3} 
			d\Qpro_t&=\alphaqq(\Qpro_{t})\,
			dt+\int\nolimits_{\mathbb R^{\nAktien}} \gammaq(\Qpro_{t-})\;
			\komppoi(dt,du),~~~
		\end{align}
		with initial values $\Mpro_0=\filterinitial$,	$\Qpro_0=\condcovinitial~$   where for $u\in\R^d$ and $q\in \R^{d\times d}$	
		\begin{align}
			\label{coeff_MQ_def}  
			\begin{array}{rlrl}				
				\betam(\variance)& =\variance\Sigma_{\HR}^{-1/2},  & 
				\gammam(\variance,u)& = \phantom{-}\variance\left(\variance+\varianceexp\right)^{-\frac{1}{2}}u,
				\\[1ex]
				\alphaqq(\variance) &= \Sigma_{\mu}-\revspeed \variance -\variance\revspeed^{\top}-\variance \Sigma_{R}^{-1} \variance
				+ \lambda\gammaq(\variance), \hspace*{1em} &				 
				\gammaq(\variance)  &=-\variance\left(\variance+\varianceexp\right)^{-1}\variance.
			\end{array}				
		\end{align}	
	\end{lemma}

	\section{Utility Maximization}
	\label{Utility_Max}
	In this section  we  investigate optimization problem  \eqref{opti_org}, i.e.~the maximization the expected power utility of the $\HC$-investor's terminal wealth. 	We will apply a change of measure technique which was already used
	among others by Nagai and Peng \cite{Nagai and Peng (2002)} and Davis and Lleo \cite{Davis and Lleo (2013_1)}. This allows to  rewrite the performance criterion and to decompose the dependence on the wealth and the filter processes. Based on this decomposition we will derive a simplified stochastic optimal control problem with a risk-sensitive performance criterion. Here, the state variables are reduced  to the  filter processes of conditional mean  and covariance whereas the wealth process is no longer contained. 				
	This problem can be treated as a stochastic optimal control problem and solved using dynamic programming methods. 
	
	\paragraph{Performance criterion} 
	Recall  equation
	\eqref{wealth_phys} for the wealth process $\wealth^{\pi}$ saying that ${dX_t^{\pi}}/{X_t^{\pi}}=
	\pi_t^{\top}dR_t$. Rewriting SDE \eqref{ReturnPro} for the return process $R$  in terms  of the
	innovations process $\widetilde{W}$ given in
	Lemma \ref{filter_C} we obtain 
	the $\mathbb F^{\HC}$-semi-martingale decomposition of $\wealth^{\pi}$
	(see also Lakner \cite{Lakner (1998)}, Sass Haussmann
	\cite{Sass and Haussmann (2004)})  
	\begin{align}
		\label{wealth-innovation}
		\frac{d\wealth_t^{\pi}}{\wealth_t^{\pi}}=\pi^{\top}_td\Mpro_t+\pi_t\sigma_{\wealth}\;d\widetilde{W}_t,\quad\wealth_0^{\pi}=x_0,
	\end{align}
	where $\sigma_{\wealth}= \Sigma_R^{1/2}$.
	From the above wealth equation we obtain that for a given admissible
	strategy  $\pi$ the power  utility of terminal wealth $\wealth_T^{\pi}$ is
	given by
	\begin{align}
		\label{Power_Nutzen_zerlegung}
		\utility_{\theta}(\wealth_T^{\pi})&=\frac{1}{\theta}(\wealth_T^{\pi})^{\theta}
		=\frac{x_0^{\theta}}{\theta} \Radon_T\exp\Big\{\int\nolimits_0^T
		{b}(\Mpro_s,\pi_s)ds \Big\}
	\end{align}
	where  for $\filter,p\in\R^d$
	\begin{align}
		\label{b_underline}
		{b}(\filter,\pointp)&=\theta\Big(\pointp^{\top}\filter-\frac{1-\theta}{2} \|\sigma_{\wealth}\pointp\|^2 \Big) 
		\quad\text{and}\\
		\label{Radon}
		\Radon_{t}&=\exp\Big\{ \theta\int\nolimits_0^t \pi_s^{\top}\sigma_{\wealth} \;d\widetilde{W}_s 
		-\frac{1}{2}\theta^2\int\nolimits_0^T \|\sigma_{\wealth}\pi_s \|^2  \;ds \Big\}, \quad  t\in[0,T].
	\end{align}
	The above process $\Radon$ is known as  Doléans-Dade exponential. It is  a martingale if 
	the admissible strategy $\pi$ is such that  $\E[\Radon_T]=1$.  Then 
	we can define an equivalent probability measure $\overline{\P}$
	by $\Radon_T=\frac{d\overline{\P}}{d\P~}$ for which Girsanov's
	theorem
	guarantees that the process
	$\overline{W}=(\overline{W}_t)_{t\in[0,T]}$ with			
	\begin{align}
		\label{Girsanov_W_prozess}
		\overline{W}_t=\widetilde{W}_t-\theta\int_0^t
		\sigma_{\wealth}\pi_s \; ds,\quad t\in[0,T],
	\end{align}
	is a $(\mathbb F^{\HC},\overline{\P})-$standard Brownian motion.
	This change of measure allows to rewrite the performance criterion $\rewardorigin_0(x_0;\pi) =
	\E\left[\utility_{\theta}(X_T^{\pi})~|~\mathcal{F}^{\HC}_0\right]$ of the  utility maximization problem \eqref{opti_org} for $\pi\in\mathcal
	A_0$ as
	\begin{align}\nonumber		
		\rewardorigin_0(x_0;\pi)&= \frac{x_0^{\theta}}{\theta} \;
		\E\Big[  \Radon_T\exp\Big\{\int\nolimits_0^T {b}(\Mpro_s^{\filterinitial,\condcovinitial},\pi_s)ds \Big\}\Big]\\
		\label{ExpUtility_RiskSensitive}
		&=\frac{x_0^{\theta}}{\theta} \;\E^{\overline{\P}}\Big[\exp\Big\{\int\nolimits_0^T {b}(\Mpro_s^{\filterinitial,\condcovinitial},\pi_s)ds \Big\}\Big].
	\end{align}
	The above used notation $\Mpro^{\filterinitial,\condcovinitial}$ emphasizes the dependence of the filter process $\Mpro$ on the initial values $\filterinitial,\condcovinitial$ at time $t=0$  and reflects the conditioning in $\E\left[\utility_{\theta}(X_T^{\pi})~|~\mathcal{F}^{\HC}_0\right]$ on the initial information given by the $\sigma$-algebra $\mathcal{F}^{\HC}_0$. 
	It turns out that the utility maximization problem \eqref{opti_org}
	is equivalent to the maximization of 
	\begin{align}\label{opti-equivalent}
		\Jpi(\filterinitial,\condcovinitial;\pi)=\E^{\overline{\P}}\Big[\exp\Big\{\int\nolimits_0^T
		{b}(\Mpro_s^{\filterinitial,\condcovinitial},\pi_s)ds \Big\}\Big]
	\end{align}
	over all admissible strategies for $\theta\in(0, 1)$ while for
	$\theta< 0$ the above expectation has to be minimized.   In the literature \eqref{opti-equivalent} is often called risk-sensitive performance criterion,  e.g. in  Davis and Lleo \cite{Davis and Lleo (2013_2)}, Nagai and Peng \cite{Nagai and Peng	(2002)}. Note that it holds $\rewardorigin_0(x_0;\pi)=\frac{1}{\theta}x_0^{\theta}\, \Jpi(\filterinitial,\condcovinitial;\pi)$.
	This allows us to remove the wealth process $X$ from the state of the control problem which we formulate next.

	\paragraph{State process}  Recall the $\P$-dynamics of the filter processes
	$\Mpro$ and $\Qpro$   given in Lemma \ref{Filter_C_Int}. Using \eqref{Girsanov_W_prozess} the dynamics under $\overline{\P}$  are obtained by expressing $\widetilde{W}$ in terms of $\overline{W}$ and leads to  the SDEs
	\begin{align}
		\label{Filter_M_int_Ph}
		d\Mpro_t& =\alphamm(\Mpro_t,\Qpro_t,\pi_t)\, dt 
		+\betam(\Qpro_{t})\,d\overline{W}_t
		+\int\nolimits_{\mathbb R^{\nAktien}} \gammam(\Qpro_{t-},u)\; \komppoi(dt,du),	\\
		\label{Filter_Q_int_Ph} 
		d\Qpro_t&=\alphaqq(\Qpro_{t})\,dt
		~~+~\int\nolimits_{\mathbb R^{\nAktien}} \gammaq(\Qpro_{t-})\;	\komppoi(dt,du),
	\end{align}
	with initial values $\Mpro_0=\filterinitial$ and $\Qpro_0=\condcovinitial$ where  for $m,p\in\R^d$ and $q\in \R^{d\times d}$ 
	\begin{align}
		\label{alphamm_def}
		\alphamm(\filter,\variance,\pointp)=\revspeed(\revlevel-\filter)+ \theta\variance\pointp,
	\end{align}
	and $\alphaqq, \betam,\gammam,\gammaq$ are given in \eqref{coeff_MQ_def}.  
	Note that the drift coefficient $\alphamm$ in the SDE \eqref{Filter_M_int_Ph} for $\Mpro$ now depends also on the conditional covariance $\Qpro$ as well as on the strategy $\pi$. The dynamics of the conditional covariance $\Qpro$  is not affected by the change of measure, so SDE  \eqref{Filter_Q_int_Ph} coincides with SDE \eqref{Riccati_CN_3}.

	For the  treatment of the optimization of $\Jpi(\filterinitial,\condcovinitial;\pi)$ given in \eqref{opti-equivalent} as a stochastic optimal control problem we have to take into account that the dynamics $\Mpro$ appearing in the performance criterion depends via its diffusion and jump coefficients in \eqref{Filter_M_int_Ph} on the conditional covariance $\Qpro$.  Although the dynamics of $\Qpro$ is  deterministic  between the arrival dates $T_1,T_2,\ldots $ of the expert opinions,  these dates defining the jump times are random. Therefore,  $\Qpro$ is a (piecewise deterministic) stochastic process and  has  to be included in the state of the stochastic control problem. 
	Thus,  the state variable now  formally reads as $(\Mpro,\Qpro)$ and takes values in  $\R^{d+d ^2}$.
	The  first state component $\Mpro$ is a $\nAktien$-dimensional vector while  the  second component $\Qpro$
	is a $\nAktien\times\nAktien$-matrix.

	In order to apply results from the theory of controlled jump-diffusions which are available for vector-valued state processes we recast the SDE for  the conditional covariance matrix  $\Qpro$ which is known to be symmetric. For such matrices it is enough to know the  entries on and below  the diagonal of $\Qpro$ which we want to collect in a  vector $\YQ$ of dimension $\nQ=d(d+1)/2$. 	
	We  define the mapping  $\restri:\mathbb R^{\nAktien\times\nAktien}\rightarrow\mathbb R^{\nQ},  ~q \mapsto \restri(q)=g$ restricting the $d\times d$ symmetric matrix $q$ to the $\nQ$-dimensional vector $g$ by 
	\begin{align}
		g:=\restri(q)=\big(q^{11}, q^{21},q^{22},q^{31},\ldots,q^{\nAktien \nAktien}  \big)^{\top},
		\label{transf_Y_QQ}
	\end{align}
	meaning that the matrix entry $q^{ij}~$ with $~ i\geq j ~$ is mapped to the entry of the vector $g$ at the position $k=K(i,j)=\frac{1}{2}i(i-1)+j$.			
	Conversely, we can reconstruct the symmetric matrix $q$ from the vector $g$ using  the inverse mapping $\restri^{-1}:\mathbb R^{\nQ}\rightarrow\mathbb R^{\nAktien\times\nAktien}, ~g \mapsto \restri^{-1}(g)=q$ which can be  defined element-wise $\text{for } k=1,\ldots, \nQ$ as		
	\begin{equation}
		\label{ruecktransform_Q}
		\begin{split}	
			g^k \mapsto q^{ij}=q^{ji} \quad  \quad \text{with}\quad 
			i&=\max\big\{ l\in\{1,\ldots,\nAktien\},~ K(l,1)\leq k \big\},	~~								
			j=k-K(i,1)+1. 
		\end{split}
	\end{equation}				
	The above mapping allows us to restrict the process  $(\Mpro,\Qpro)$ of dimension $d+d^2$  with the matrix-valued second component $\Qpro$ to  the vector-valued process  $\stateprocess= {\Mpro \choose \YQ}$ of dimension  $\nZustand=\nAktien+\nQ=d(d+3)/2$ which in the following will serve  as state process. The associated  state space is denoted  by $\Statespace^{\stateprocess}\subset \mathbb R^{\nZustand}$.   It can be decomposed as $\Statespace^{\stateprocess}=\mathcal S^{\Mpro}\times\mathcal S^{\YQ}$ with the domain  $\mathcal S^{\Mpro}=\mathbb R^{\nAktien}$ of the first component $\Mpro$, and the domain 	$\mathcal{S}^{\YQ}  \subset \mathbb R^{\nQ}$ of second component $\YQ$.  Note that $\mathcal{S}^{\YQ}$ can be chosen as a bounded subset of $\mathbb R^{\nQ}$ since Prop.~\ref{properties_filter} on the boundedness of the conditional covariance  $\Qpro$ implies that there is a constant $K_G>0$ such that  $\mathcal{S}^{\YQ} \subset \{\zsmall \in \R^{\nQ}: \|g\|\le K_G\}$.  In order to present the SDEs governing the dynamics of the vector process $\YQ$ and the state process $\stateprocess$   we redefine the drift and jump coefficients  in SDE \eqref{Filter_Q_int_Ph} in terms of the vector $\YQ$ instead of the matrix  $\Qpro$ using the mapping $\variance=\mathcal R^{-1}(g)$  as follows						
	\begin{align}\nonumber
		\alphaYQ(\zsmall):=\restri(\alphaqq(\variance)), \quad
		\gammaYQ(\zsmall):=\restri(\gammaq(\variance)). 			
	\end{align}
	Then, the dynamics of the restricted state process $\stateprocess$  for a fixed  admissible strategy $\pi\in\mathcal
	A_0$  are given by 
	\begin{align}
		\label{State_SDE}
		d\stateprocess_t& =\alphay(\stateprocess_t,\pi_t) \;dt+\betay(\stateprocess_t)\;  dW_t+\int\nolimits_{\mathbb R^{\nAktien}} \gammay(\stateprocess_{t-},u) \komppoi(dt,du),~~
		\stateprocess_0 ={\filterinitial \choose\restri(\condcovinitial)},			
	\end{align}
	where the coefficients $\alphay$, $\betay$ and $\gammay$
	are defined by
	\begin{align}
		\alphay(y,p):=\begin{pmatrix} \alphamm(m,\variance,p)\\  \alphaYQ(\zsmall)  \end{pmatrix},~~
		\betay(y):=\begin{pmatrix} \betam(\variance) \\  0_{\nQ \times \nAktien} \end{pmatrix},~~
		\gammay(y,u):=\begin{pmatrix} \gammam(\variance,u) \\ \gammaYQ(\zsmall)  \end{pmatrix},\label{coeffs_rest_state}
	\end{align}
	with $y= {m\choose \zsmall}=\mathbb R^{\nZustand}$ and $\variance=\mathcal{R}^{-1}(\zsmall)$.

	For expressing the	performance criterion given in \eqref{opti-equivalent} in terms of
	the state variable $Y$, we rewrite with some abuse of notation  the function $b$	defined in \eqref{b_underline} for $	y={\filter \choose g}$ as			
	\begin{align}
		\label{function_b}
		b(y,p)=b((\filter,g),p):=\theta\Big(\pointp^{\top}\filter-\frac{1-\theta}{2}\|\sigma_{\wealth}\pointp\|^2\Big).
	\end{align}
	Then the performance criterion from \eqref{opti-equivalent} for a
	given strategy $\pi$ reads as
	\begin{align}\label{opti-equivalent-b}
		\Jpi(\pi):=\E^{\overline{\P}}\Big[  \exp\Big\{\int\nolimits_0^T
		b(\stateprocess_s^{ y_0},\pi_s)ds \Big\}\Big],
	\end{align}
	where $\stateprocess^{y_0}$ denotes the state process starting at time $t=0$ with the initial value $y_0={\filterinitial \choose \zsmall_0}$ with $\zsmall_0=\restri(\condcovinitial)$.
	
	\medskip
	\paragraph{Markov Strategies}
	For applying the dynamic programming approach to the optimization
	problem \eqref{opti-equivalent-b} the state process $\stateprocess$
	needs to be a Markov process adapted to $\mathbb F^{\HC}$. To allow
	for this situation we restrict the set of admissible strategies to
	those of Markov type which are defined in terms of the state process
	$\stateprocess$  by a  decision rule $\Pi$ via	$\pi_t=\Pi(t,\stateprocess_t)$ for  some given measurable
	function $\Pi : [0,T]\times\Statespace^{\stateprocess}\rightarrow
	\mathbb{R}^{\nAktien}$. 
	Below  we will need some
	technical conditions on $\Pi$ which we collect in the following
	\begin{assumption}\label{admi_stra_rule}\hspace{20cm}
		\begin{enumerate}
			\item {\textbf{Lipschitz condition and linear growth condition}   \\
				There exists  constants $C_M,C_G>0$ such that for all $y, y_1,
				y_2\in\Statespace^{\stateprocess}$ and all $t\in[0,T]$ it holds
				\begin{align}\label{Lipschitz_rule}
					\|\Pi(t,y_1)-\Pi(t,y_2)\|& \leq C_M\|y_1-y_2\|, \\[0.5ex]
					\label{Linear-growth-rule}
					\|\Pi(t,y)\|& \leq C_G(1+\|y\|).
			\end{align}}
			\ \\[-6ex]
			\item \textbf{Integrability condition} \\ 
			The strategy processes $\pi$ defined by $\pi_t=\Pi(t,\Mpro_t)$ on $[0,T]$  are such that the  process $\Radon$ defined by \eqref{Radon} satisfies 	$\E[\Radon_T]=1$.

		\end{enumerate}
	\end{assumption}
	We denote by  
	\begin{align}
		\label{set_admiss_Markov} 
		\hspace{-0.3cm}\mathcal{A}:=\big\{\Pi:[0,T]\times \Statespace^{\stateprocess}\to\R^d: \Pi \text{ is a measurable function satisfying Ass.~\ref{admi_stra_rule}}\big\}
	\end{align}
	the \textit{set of admissible decision rules}.
	\begin{remark}
		\label{remark on decision rule} 			
		The integrability  condition 
		guarantees that the Radon-Nikodym density 	process $\Radon$ given in \eqref{b_underline} is 
		a martingale, hence	the equivalent measure	$\overline{\P}$ is well-defined.
		A Markov strategy $\pi=(\pi_t)_{t\in[0,T]}$ with
		$\pi_t=\Pi(t,\Mpro_t)$ defined by an admissible decision rule
		$\Pi$ is contained in the set of admissible strategies $\mathcal
		{A}_0$ given in \eqref{set_admiss_0} since by construction it is
		$\mathbb{F}^{\HC}$-adapted, the positivity of the wealth process
		$\wealth^{\pi}$ follows from \eqref{Power_Nutzen_zerlegung}. The integrability condition implies the
		square-integrability of $\pi$.  Finally, the Lipschitz and linear growth condition ensure that SDE \eqref{State_SDE} for the dynamics for the controlled state process  enjoys for all admissible strategies a unique strong solution.
	\end{remark}

	\textbf{Control problem}. We are now ready to formulate the
	stochastic optimal control problem with the state process $\stateprocess$ and a Markov control defined by the decision rule
	$\Pi$. Recall that the dynamics of the state process is now given in \eqref{State_SDE} with initial value
	$\stateprocess_0=  {\filterinitial \choose  \zsmall_0}$ with {$\zsmall_0=\restri(\condcovinitial)$}.  We now vary the initial condition and start the processes at time $t\in[0,T]$ with initial value $\stateprocess_t=y= {m \choose g}$, i.e.,  $\Mpro_t=m$ and $\YQ_t=g$.   Let us denote  the state process at time
	$s\in[0,T]$ to a fixed decision rule $\Pi$ by $\stateprocess_s^{ \Pi,t,y}$ and its components by $\Mpro_s^{  \Pi,t,y}$ and
	$\YQ_s^{ \Pi,t,y}$. Note that $\YQ^{ \Pi,t,y}$ only
	depends on  the initial value $ \zsmall$ but does not  depend on $m$ and $\Pi$.
	
	For solving the control problem \eqref{opti-equivalent-b} we will apply the
	dynamic programming approach which requires the introduction of the
	following reward and value functions. For all $t\in[0,T]$ and
	$\Pi\in\mathcal {A}$ the associated reward function of the control
	problem \eqref{opti-equivalent-b} reads as
	\begin{align}
		\label{Zielfkt_H} \reward(t,y;\Pi):&= \E^{\overline{\P}}\Big[  \exp\Big\{
		\int\nolimits_t^Tb(\stateprocess_s^{   \Pi,t,y},\Pi(s,\stateprocess_s^{  \Pi,t,y}))ds
		\Big\}\Big],\quad\text{for}\quad \Pi\in\mathcal{A},
	\end{align}
	while the associated value function reads as
	\begin{align}
		\label{Wertfkt_H}
		\valuefkt(t,y):&=\begin{cases}\sup\limits_{\Pi\in\mathcal A}\reward(t,y;\Pi),&\quad \theta\in(0,1), \\[2ex] \inf\limits_{\Pi\in\mathcal A}\reward(t,y;\Pi),&\quad \theta\in(-\infty,0), \end{cases}
	\end{align}
	and it holds $\valuefkt(T,y)=\reward(T,y;\Pi)=1$. In the sequel we
	will concentrate on the case $\theta\in(0,1)$, the necessary changes
	for $\theta < 0$ will be indicated where appropriate.
	In view of relation \eqref{ExpUtility_RiskSensitive}  and the above transformations the value function of the original utility maximization problem \eqref{opti_org} can be obtained  from the value function of control problem \eqref{Wertfkt_H} by 		
	\begin{align}
		\label{Value_Orig_RiskSens}
		\mathcal{V}(x_0)=\frac{x_0^\theta}{\theta} \valuefkt(0,y_0) \quad\text{with }\quad  y_0= (\filter_0,\condcovinitial).
	\end{align}

	\section{Utility Maximization for Partially Informed Investors With Expert Opinions	}
	\label{UtilityMaxDiscreteExperts} 
	
	In this section we want to derive the dynamic programming equation for the control problem given above in \eqref{Wertfkt_H} which consists in the optimization of the risk-sensitive performance criterion or reward function \eqref{Zielfkt_H}. Recall from \eqref{State_SDE} and \eqref{coeffs_rest_state}    
	the dynamics of the  state process $\stateprocess= {\Mpro \choose\YQ}$. For deriving the dynamic programming equation for $V(t,y)$ we first
	need to state the generator of  the state process
	$\stateprocess$. 
	
	\begin{lemma}
		\label{Generator_pi}  The generator
		$\generator=\generator^{p}$ of the state process
		$\stateprocess$ given in equation
		\eqref{State_SDE} for $f\in C^2(\Statespace^{\stateprocess})$  is given for a fixed control $\pointp\in\R^\nAktien$ and $y={\filter \choose \zsmall }  \in \Statespace^{\stateprocess}$ by
		\begin{align}
			\mathcal{L}^{\pointp}f(y)=&\diffop_{\filter}^{\top}f(y)\alphamm(y,\pointp)
			+\diffop_{\zsmall}^{\top} f(y)\alphaYQquer(\zsmall)\nonumber +\frac{1}{2}\trace\big\{\diffop_{\filter\filter} f(y)\betam(\variance)\betam^{\top}(\variance)\big\}\nonumber\\
			&\label{generator_restri}+ \lambda
			\Big\{\int_{\mathbb R^{\nAktien}}f\big(\filter+\gammam(\variance,u),~\zsmall+\gammaYQ(\zsmall)\big)\varphi(u)du-f(y)
			\Big\},
		\end{align}
		where   $\alphaYQquer(\zsmall):=\alphaYQ(\zsmall)-\lambda\gammaYQ(\zsmall)$,
		$\variance=\mathcal{R}^{-1}(\zsmall)$, $\diffop_{\filter} f,  \diffop_{\filter\filter} f$ denoting the gradient and Hessian matrix of $f$ w.r.t.~$\filter$ and  $\diffop_{\zsmall} f$ the gradient of $f$ w.r.t.~$\zsmall$, respectively.
		
	\end{lemma}
	
	\begin{proof}
		Standard arguments show that the solution of equation
		\eqref{State_SDE} is a Markov process whose generator
		$\mathcal{L}^{\pointp}$ operates on $f\in \mathcal{C}^2
		(\Statespace^{\stateprocess})$ as follows
		\begin{align}
			\mathcal{L}^{\pointp}f(y)=&\diffop_{\filter}^{\top}f(y)\alphamm(\filter,\variance,\pointp) +
			\diffop_{\zsmall}^{\top} f(y)\alphaYQ(\zsmall)\nonumber+\frac{1}{2}\trace\big\{\diffop_{\filter\filter} f(y)\betam(\variance)\betam^{\top}(\variance)\big\}\nonumber\\
			&+\lambda
			\int_{\mathbb R^{\nAktien}}\big(f\big(\filter+\gammam(\variance,u)~,~\zsmall+
			\gammaYQ(\zsmall)\big)-f(y)\big)   \varphi(u)du\nonumber\\
			&-\lambda
			\int_{\mathbb R^{\nAktien}}\big( \diffop_{\filter}^{\top} f(y)\gammam(\variance,u)+
			\diffop_{\zsmall}^{\top} f(y)\gammaYQ(\zsmall)\big)   \varphi(u)du.
		\end{align}				
		Recall that $\gammam(\variance,u)=\variance(\variance+\Gamma)^{-1/2}u$ which implies that $\int_{\mathbb R^{\nAktien}}\gammam(\variance,u)\varphi(u) du=0$ so that it holds  for the last term 
		\begin{align}
			\int_{\mathbb R^{\nAktien}}\big( \diffop_{\filter}^{\top} f(y)\gammam(\variance,u)+\diffop_{\zsmall}^{\top} f(y)\gammaYQ(\zsmall)\big)\varphi(u) du
			&=\diffop_{\zsmall}^{\top} f(y)\gammaYQ(\zsmall)\nonumber
		\end{align}
		and  together with the second  term  we obtain
		\begin{align}
			\diffop_{\zsmall}^{\top} f(y)\alphaYQ(\zsmall)-\lambda \diffop_{\zsmall}^{\top} f(y)\gammaYQ(\zsmall)
			=\diffop_{\zsmall}^{\top} f(y)\alphaYQquer(\zsmall).
			\nonumber
		\end{align}
		\qed
	\end{proof}
	For solving control problem \eqref{Wertfkt_H} we now apply dynamic programming techniques. Starting point is the dynamic programming principle given in the following lemma. For the proof we refer to Frey et al.~\cite[Prop.~6.2]{Gabih et al (2014)} and  Pham~\cite[Prop.~3.1]{Pham (1998)}.

	\begin{lemma}(Dynamic Programming Principle).
		\label{lemma_DPP}\\
		For every $t\in[0,T]$, $y={\filter \choose \zsmall}\in\mathbb R^{\nZustand}$ and for
		every stopping time $\tau$ with values in $[t,T]$ it
		holds
		\begin{align}
			\valuefkt(t,y)=
			\sup\limits_{\Pi\in\mathcal A} \E^{\overline{\P}}  \Big[\exp\Big\{\int_t^{\tau} b\big(\stateprocess_s^{\Pi,t,y},\Pi(s,\stateprocess_s^{\Pi,t,y})\big)ds\Big\}
			\valuefkt(\tau,\stateprocess_{\tau}^{\Pi,t,y}) \Big]. \label{DPP_Y}
		\end{align}
	\end{lemma}
	Next we turn to the derivation of the dynamic programming equation (DPE)  for the  value function $\valuefkt(t, y)$ of  optimization problem \eqref{Wertfkt_H}.  For this purpose we suppose  that the value function $\valuefkt$ is a $C^{1,2}$-function.  Then Theorem \ref{DPE_general} below shows that the DPE  for the current problem  appears as a non-linear partial integro-differential equation (PIDE). It  constitutes a necessary optimality condition  and allows to derive the optimal decision rule.  We recall that we focus on  the solution for $\theta\in(0,1)$, the case  $\theta <0$ follows analogously by changing  $\sup$  into $\inf$ in  \eqref{DPP_Y}. 
	
	\begin{theorem}[Dynamic programming equation]%
		\label{DPE_general}~
		\begin{enumerate} 
			\item The value function $\valuefkt$ satisfies  for  $t\in[0,T)$ and  $y={\filter \choose \zsmall}\in\mathbb R^{\nZustand}$ the PIDE 
			\begin{align}
				\label{HJB_Power_C} 
				\begin{split}
					0=&~~~~\frac{\partial}{\partial t}
					\valuefkt(t,y)+\diffop_{\filter}^{\top}
					\valuefkt(t,y)\big(\revspeed(\revlevel-\filter)+\frac{\theta}{1-\theta}\Sigma_{R}^{-1}\variance
					\filter \big) +\diffop_{\zsmall}^{\top}\valuefkt(t,y)\alphaYQquer(\zsmall)						\\
					&
					+\frac{1}{2}\trace\big\{\diffop_{\filter\filter}\valuefkt(t,y)\big( \variance\Sigma_{\HR}^{-1}\variance\big)\big\}+\frac{\theta}{2(1-\theta)}\filter^{\top}\Sigma_{R}^{-1}
					\filter \valuefkt(t,y)\\
					&
					+\frac{\theta}{2(1-\theta)}\frac{1}{\valuefkt(t,y)}
					\diffop_{\filter}^{\top}\valuefkt(t,y)\big(\variance\Sigma_{R}^{-1}\variance\big)\diffop_{\filter}\valuefkt(t,y)			\\
					&						+\lambda\Big\{\int_{\mathbb
						R^{\nAktien}}\valuefkt\big(t,\filter+\gammam
					(\zsmall,u),~\zsmall+\gammaYQ(\zsmall)\big) \varphi(u)du
					-\valuefkt(t,y)\Big\},
				\end{split}
			\end{align}
			with the terminal condition $\valuefkt(T,y)=1$. 				 
			\item
			The optimal decision rule   is for $t\in [0,T)$ and  $y={\filter \choose \zsmall}\in\R^{\nZustand}$ given by 		
			\begin{align}
				\label{opti_stra_power_C}
				\begin{split}
					\Pi^{*}=\Pi^{*}(t,y)&=\frac{1}{1-\theta}\Sigma_{R}^{-1}\Big(\filter+\variance\frac{\diffop_{\filter}\valuefkt(t,y)}{\valuefkt(t,y)}  \Big).
				\end{split}						
			\end{align}
		\end{enumerate}
	\end{theorem}
	
	\begin{proof}		
		Following standard arguments for controlled jump diffusions  maximizing a risk-sensitive performance criterion as in Frey et al.~\cite{Frey et al. (2012)} the above dynamic programming principle implies that the value function satisfies a DPE of the form		
		\begin{align}\label{eq:HJB-general}
			\frac{\partial}{\partial t}	V(t,y)+\sup_{p \in \mathbb	R^{\nAktien}}\big\{\mathcal{L}^{p}\; 	V(t,y)-b(y,p) V(t,y)\big\}& =0, ~(t,y) \in [0,T) \times \R^{\nZustand},\\[-1ex]
			V(T,y) & =1.		
		\end{align}
		Expressing in the above equation the generator  $\mathcal{L}^{p}$ as given in Lemma \ref{Generator_pi} with $\alphamm$ as given in \eqref{alphamm_def} and $ b(y, p)$ as given
		in \eqref{function_b}, the DPE can be written more explicitly as the following PIDE	
		\begin{align}
			0&=\frac{\partial}{\partial t}
			\valuefkt(t,y)+\diffop_{\filter}^{\top}	\valuefkt(t,y)\revspeed(\revlevel-\filter)
			+\diffop_{\zsmall}^{\top}\valuefkt(t,y)\alphaYQquer(\zsmall)
			+\frac{1}{2}\trace\big\{\diffop_{\filter\filter}\valuefkt(t,y)	\variance\Sigma_{R}^{-1}\variance\big\}
			\nonumber\\
			&+\sup\limits_{\pointp\in\R^{\nAktien}}\big\{\diffop_{\filter}^{\top}	\valuefkt(t,y)
			\theta\variance\pointp+\theta\big(\pointp^{\top}\filter -\frac{1-\theta}{2}\pointp^{\top}\Sigma_{R}\pointp\big) \valuefkt(t,y)\big\}
			\label{Bell_equa_power_C}\\
			&+
			\lambda\Big\{\int_{\mathbb
				R^{\nAktien}}\valuefkt\big(t,\filter+\gammam
			(\zsmall,u)~,~\zsmall+\gammaYQ(\zsmall) \big)\varphi(u)du
			-\valuefkt(t,y)\Big\}.
		\end{align}
		The maximizer for the supremum appearing
		in \eqref{Bell_equa_power_C} yields the 	optimal decision rule $\Pi=\Pi(t,y)$ which is given in \eqref{opti_stra_power_C}. 
		Plugging the expression for  the maximizer $\Pi$ back into the DPE
		\eqref{Bell_equa_power_C}  we obtain the  PIDE \eqref{HJB_Power_C}. 
		\qed
	\end{proof}
	
	\begin{remark}\label{myopic}
		The $\HC$-investor's optimal decision rule given in \eqref{opti_stra_power_C} can be rewritten as 
		\begin{align}
			\label{opti_stra_power_myopic}
			\Pi^{*}(t,y)& =
			\Pi^{\HF}(t,\filter)+\frac{1}{1-\theta}\Sigma_{R}^{-1}\variance\frac{\diffop_{\filter}\valuefkt(t,y)}{\valuefkt(t,y)},
		\end{align}
		with $\Pi^{F}=\frac{1}{1-\theta}\Sigma_{R}^{-1}\filter$. Similarly to the setting with expert opinions arriving at fixed points in time which was studied in Gabih et al.~\cite{Gabih et al (2023) PowerFix} we find out that  $\Pi^{*}$ can be decomposed into two parts: The first part  $\Pi^{F}$ is the optimal decision rule of the fully informed investor, and a second part $\Pi^{*}- \Pi^{F}$ depending on filtering-related quantities. In case of fully observed drift  the value of the strategy process $\pi_t$ at time $t$ is obtained by plugging in for $m$  the current value of the drift $\drift_t$, whereas  the partially informed $\HC$-investor plugs in the filter estimate $\Mpro_t$. 	
		In the literature  $\Pi^{F}$ is also known as  {\em myopic decision rule}. The second part is  the ``correction term''  $\Pi^{*}- \Pi^{F}$ which is known as  {\em drift risk} of the partially informed $\HC$-investor since it accounts for the investor's uncertainty about the current value of the non-observable drift, see Rieder and Bäuerle \cite[Remark 1]{Rieder_Baeuerle2005} and Frey et al.~\cite[Remark 5.2]{Frey et al. (2012)}. Further, The decomposition shows that, in contrast to the case of log utility (see Sass et al.~\cite{Sass et al (2017)}), the so-called  {\em certainty equivalence principle} does not apply to power utility. 		 	It states  that the optimal strategy under partial information
		is obtained by replacing the unknown drift $\drift_t$ by the filter estimate $\Mpro^H_t$ 	in the formula for the optimal strategy under full information. 
	\end{remark}

	\section{Regularization Approach}
	\label{regularization}
	The major challenge in the analysis   of the DPE in form of the PIDE \eqref{HJB_Power_C} lies in the
	fact that its diffusion part is degenerated so that it is not
	possible to apply any of the known results on the existence of
	classical solutions to this equation, and therefore the existence of a classical solution of equation \eqref{HJB_Power_C} is an open issue.
	The main problem is the fact
	that one cannot guarantee that the differential operator is uniformly elliptic.
	To see this, note that  the SDE \eqref{State_SDE} for the controlled state process of dimension $\nZustand =d+\nQ$ is driven by a Brownian motion  of dimension $d$ only.  Thus, there is a singular coefficient matrix of the second
	derivatives in \eqref{HJB_Power_C} which is given by  $A(y)=
	\betay(y)\betay^{{\top}}(y)$ with  $\betay(y)={\betam( \restri^{-1}(\zsmall)) \choose 0_{\nQ \times \nAktien}}$ as given in \eqref{coeffs_rest_state}. By
	definition, equation \eqref{HJB_Power_C} is uniformly elliptic if we
	can find some constant $C>0$ such that
	\begin{align}
		z^{\top} A(y) z
		>C\|z\|^2,\quad \text{for all} ~z\in\mathbb R^{\nZustand}\setminus\{0\}~\text{and}~\text{for all}~y\in\mathbb R^{\nZustand}\nonumber,
	\end{align}
	in particular the matrix $A$ needs to be strictly positive definite. To address this fact we follow
	the regularization procedure  which was  used recently in Frey et al. \cite{Gabih
		et al (2014)} and Shardin and Wunderlich \cite{Shardin and Wunderlich (2017)}  and also earlier in Fleming and Soner \cite[Lemma IV.6.3]{Fleming and Soner(2006)} and Krylov \cite[Section 4.6]{Krylov1980}. However, in contrast to the setting  of a drift driven by a finite-state
	Markov chain considered in \cite{Gabih et al (2014),Shardin and Wunderlich (2017)}, the picture changes considerably in our setting as we deal with an unbounded
	Ornstein-Uhlenbeck drift process which creates technical
	difficulties when proving the convergence of the regularized value
	functions. The announced regularization argument  will show
	that \eqref{HJB_Power_C} can be used to compute an approximately
	optimal strategy. For this we add  to the dynamics of the state equation
	\eqref{State_SDE} a term			$\frac{1}{\sqrt{k}}dW^{\ast}$ with $k\in\N$ and $W^{\ast}$ is a
	$\nZustand$-dimensional Brownian motion independent of
	$W$.      Denoting the value function of the control problem with  the modified state process by ${}^\np V$, the DPE associated
	with these modified dynamics has an additional term
	$\frac{1}{2 k}\sum_{j=1}^{\nAktien+\nQ}~ {}^\np V_{y_jy_j}$ and is therefore
	uniformly elliptic. Hence the results of Davis and Lleo \cite[Theorem 4.17]{Davis and Lleo (2013_2)}
	apply directly
	to the modified equation, yielding the existence of a classical
	solution ${}^\np V$. Moreover, the optimal decision rule ${}^\np
	\Pi^{\ast}$ of the modified problem is given by \eqref{HJB_Power_C}
	with ${}^\np V$ instead of $V$. Clearly, one expects that for $k$
	sufficiently large ${}^\np \Pi^{\ast}$ is approximately optimal in
	the original problem and in fact we can show this now. Another route to for justifying equation \eqref{HJB_Power_C} is presented in
	\cite{Kondkaji (2019)} where the author presents numerical results obtained by the approximate solution of PIDE \eqref{HJB_Power_C} using finite difference schemes.
	%
	
	\subsection{\it \textbf{Regularized State Equation}}
	\label{Regularisierte_Zustandsgleichung}
	
	The announced regularization  to the diffusion part		of the state equation will drive the second state component $\YQ$ representing the entries of the covariance matrix $\Qpro$ out of its    domain $\mathcal S^{\YQ}$ which is by  Proposition \ref{properties_filter} bounded. Thus   $\YQ$ now takes values  in the new state space $\widetilde{\mathcal	S}^{\YQ}$  which is the whole $\R^{\nQ}$.   
	Accordingly, the resulting state space for the regularized state process ${}^\np\stateprocess$ is now the whole $ \mathbb R^{\nZustand}$, we denote it  by  $\widetilde{\mathcal	S}_Y$. 
	
	The loss of the boundedness property of $\YQ$  will create various technical difficulties. 
	First,  we need to extend the definition of the	drift, diffusion an jump  coefficients $\alphay$, $\betay$ and
	$\gammay$ of the original state process $\stateprocess={\Mpro \choose\YQ}$ given in \eqref{coeffs_rest_state}. 
	In order to extend the definition  of the coefficients from  $\mathcal{S}^{Y}$ to $\widetilde{\mathcal{S}}^{Y}$, i.e., to allow the second component of $y={m \choose \zsmall}$ take values in $\widetilde{\mathcal	S}^{\YQ}=\R^{\nQ}$ instead of the bounded domain ${\mathcal	S}^{\YQ}$, we introduce for   $\zsmall\in \widetilde{\mathcal S}^{\YQ}$  the 
	distance to $\mathcal S^{\YQ}$  as 
	$$\dist(\zsmall, \mathcal S^{\YQ}):=\inf\{ \|\zsmall-z\|, \quad z\in \mathcal S^{\YQ}\}.$$				
	Now we  fix some  $\varepsilon>0$ and define the $\varepsilon$-neighborhood of $\mathcal{S}^{\YQ}$ by  $$ \mathcal{S}^{\YQ}_\varepsilon=  \{\zsmall \in \R^{\nQ}: \dist(\zsmall, \mathcal S^{\YQ})\le \varepsilon\}.$$			This allows to define 			
	the extension of the coefficients for the regularized state process  				
	\begin{align}
		\widetilde{\alpha}_Y(y,p):=\begin{pmatrix} \widetilde{\alpha}_M(m,\variance,p)\\  \widetilde{\alpha}_\YQ(\zsmall)  \end{pmatrix},~~
		\widetilde{\beta}_Y(y):=\begin{pmatrix} \widetilde{\beta}_M(\variance) \\  0_{\nQ \times \nAktien} \end{pmatrix},~~
		\widetilde{\gamma}_Y(y,u):=\begin{pmatrix} \widetilde{\gamma}_M(\variance,u) \\ \widetilde{\gamma}_\YQ(\zsmall)  \end{pmatrix},\label{coeffs_regu_state}
	\end{align}
	for every
	$y={m \choose \zsmall}\in \widetilde{\mathcal S}^{Y}$ with $q=\restri^{-1}(g)$ such that 
	for $\zsmall \in \mathcal{S}^{\YQ}$  the extended and original coefficients coincide.
	For $\zsmall \in \R^{d_\YQ}\setminus \mathcal{S}^{\YQ}_\varepsilon$ the extended coefficients vanish whereas  in the ``transition domain ''  $\mathcal{S}^{\YQ}_\varepsilon\setminus  \mathcal{S}^{\YQ} $ there is a continuous transition of the coefficients  to zero if $\dist(\zsmall, \mathcal S_{\YQ}) $ reaches	$\varepsilon$. For the drift coefficient this can be obtained  as follows
	\begin{align}
		\widetilde{\alpha}_Y(y,p) := \left\{
		\begin{array}{cl}  \alphay(y,p)
			\big(1-\frac{1}{\eps}\dist(\zsmall, \mathcal S^{\YQ}) \big) ,\quad &
			\text{for }~ \zsmall\in \mathcal{S}^{\YQ}_\varepsilon, \\[1ex] 0, &
			\text{otherwise}.\end{array} \right.\nonumber
	\end{align}
	Analogously we define  the  other extended coefficients.  For values of $\zsmall\in \mathcal{S}^{\YQ}_\varepsilon\setminus \mathcal{S}^{\YQ}$  we describe the dependence of the  coefficients  on $\zsmall$ by the expressions given in  \eqref{coeff_MQ_def},  \eqref{alphamm_def} for which up to now  $\zsmall\in\mathcal{S}^\YQ$ was assumed. 
	Note that  $\varepsilon>0 $ always can be chosen such that there exists the  inverse of  $\variance+\varianceexp$ with $\variance=\mathcal R^{-1}(\zsmall)$ appearing in the jump coefficient $\widetilde{\gamma}_Y$. This follows from the fact that $\varianceexp$ is (strictly) positive definite and that for $g\in \mathcal{S}^{\YQ}$ the matrix $q=\restri^{-1}(g)$ is non-negative definite. Then  $\variance+\varianceexp$ is positive definite not only for  $\zsmall \in \mathcal{S}^{\YQ}$ but also for  $\zsmall \in \widetilde{\mathcal{S}}^{\YQ}$ with $\dist(\zsmall, \mathcal S^{\YQ})<\varepsilon$ for some sufficiently small $\varepsilon$.

	We	also have to extend the definition of the decision rule $\Pi$ to the state space $\widetilde{\mathcal{S}}^Y$ and the associated  set of admissible decision rules $\mathcal{A}$. Therefore, we now define the domain of definition of $\Pi$ from $[0,T]\times{\Statespace}^{\stateprocess}$ to $[0,T]\times \widetilde{\Statespace}^{\stateprocess}$ and impose to $\Pi$	 the same conditions as in Assumption \ref{admi_stra_rule} with ${\Statespace}^{\stateprocess}$ replaced by $\widetilde{\Statespace}^{\stateprocess}$.  					
	Then we can  redefine 	the set of admissible decision rules for the regularized problem  as
	\begin{align}
		\label{set_admiss_Markov_regu} 
		\mathcal{A}:=\Big\{\Pi:[0,T]\times \widetilde{\Statespace}^{\stateprocess}\to\R^d: ~\Pi \text{ is a measurable function satisfying Ass.~\ref{admi_stra_rule}}\Big\}.					
	\end{align}
	Note that for the original control problem with the state process $\stateprocess$ taking values in ${\Statespace}^{\stateprocess}$ the above set is the same as in \eqref{set_admiss_Markov}. 	Therefore we  omit the tilde in the notation  and keep the notation	$\Pi$ and $\mathcal{A}$ also for the regularized problem.

	Next we define the dynamics of the regularized state process by perturbing the SDE \eqref{State_SDE}
	with a term	$\frac{1}{\sqrt{k}}dW^{\ast}$ where $k\in\N$ and $W^{\ast}$ is a Brownian motion of dimension $\nZustand$ independent of	$\overline{W}$. The resulting process is denoted by ${}^\np\stateprocess={{}^\np \Mpro \choose{}^\np \YQ} $.
	For an admissible  decision rule $\Pi\in \mathcal{A}$ defining the strategy $\pi$ by $\pi_t=\Pi(t,{}^\np Y_t)$ we obtain
	\begin{equation}
		\label{filter_pert} d\,{}^\np\stateprocess_t =
		\widetilde{\alpha}_Y({}^\np\stateprocess_t,\pi_t) dt +
		\widetilde{\beta}_Y({}^\np\stateprocess_t) d\overline{W}_t + \int\nolimits_{\mathbb
			R^{\nAktien}} \widetilde{\gamma}_Y({}^\np\stateprocess_{t-},u)\komppoi(dt,
		du) +\frac{1}{\sqrt{\np}} dW^{\ast}_t.
	\end{equation}
	This state process is now driven not only by $\overline{W}$ and $\komppoi$ but also by the		Brownion motion	$W^\ast$.  		
	\begin{remark}\label{regu_remark}
		The coefficients $\widetilde{\alpha}_Y,\widetilde{\beta}_Y,\widetilde{\gamma}_Y$ of the above SDE vanish whenever the second component  $ {}^\np \YQ$ of the regularized state  ${}^\np Y $ leaves $\mathcal{S}^\YQ_\varepsilon$, i.e., satisfies    $\dist({}^\np \YQ_t, \mathcal S^{\YQ})> \varepsilon$. Then the dynamics is given by  $d\,{}^\np\stateprocess_t=\frac{1}{\sqrt{\np}} dW^{\ast}_t$ and describes a standard Brownian motion scaled by ${1}/{\sqrt{\np}}$. 
		
		We also add a perturbation to the dynamics of the first state component ${}^k\Mpro$. This is necessary since the  diffusion coefficient in the original state equation given by  $\betam(\variance)=\variance\Sigma_{\HR}^{-1/2}$ is a linear function of the conditional covariance matrix $\variance$ which is  positive semi-definite but not strictly positive definite and might be singular. Hence,  regularization of the second component ${}^k\YQ$ only, but not of the first component   ${}^k\Mpro$,  cannot yet  guarantee the ellipticity of generator. 
	\end{remark}
	Below we will need the following properties of the coefficients of the regularized state equation \eqref{filter_pert}.

	\begin{lemma}[Lipschitz and linear growth conditions]
		\label{coef-model}\\
		There exists a positive constants   $\widetilde{C}_M, \widetilde{C}_G$ and there exists   a function $\overline{\rho}: \mathbb R^{\nAktien}\rightarrow\mathbb R_+$
		satisfying $\int_{\mathbb R^{\nAktien}}\overline{\rho}^2(u)\varphi(u)du< \infty$, so that for every $ y, y_1, y_2\in\mathbb R^{\nZustand}$ it holds	
		\begin{align}
			\label{lipsch:b,sigma}
			\|\widetilde{\alpha}_Y(y_1,p)-\widetilde{\alpha}_Y(y_2,p)\|+\|\widetilde{\beta}_Y(y_1)-\widetilde{\beta}_Y(y_2)\|&\leq \widetilde{C}_M \|y_1-y_2\| ,\\
			\label{growth_cond}
			\|\widetilde{\alpha}_Y(y,p)\|+\|\widetilde{\beta}_Y(y)\| &\leq \widetilde{C}_G(1+\|y\|),\\
			\|\widetilde{\gamma}_Y(y_1,u)-\widetilde{\gamma}_Y(y_2,u)\|&\leq \overline{\rho}(u)\|y_1-y_2\|, \label{lipsch:Delta}\\
			\|\widetilde{\gamma}_Y(y,u)\|&\leq\overline{\rho}(u)(1+\|y\|).\label{growth-delta}
		\end{align}
	\end{lemma}
	\begin{proof}
		It is sufficient to provide the proof for $y={\filter \choose \zsmall}$ from the state space $\mathcal{S}^\stateprocess=\mathcal{S}^\Mpro\times\mathcal{S}^\YQ$ of the original state process $\stateprocess={\Mpro \choose\YQ}$  and for $\zsmall\in \mathcal{S}^\YQ_\varepsilon\setminus\mathcal{S}^\YQ$ because when multiplying these functions by the bounded and Lipschitz continuous
		function $\big(1-\frac{1}{\eps}\dist(\zsmall, \mathcal S^{\YQ}) \big) \wedge 0$ the Lipschitz and growth conditions are preserved. 
		The claims in \eqref{lipsch:b,sigma} and in 
		\eqref{growth_cond}  are clear since the conditional covariance $\Qpro$ takes values  in a bounded domain as stated by Proposition \ref{properties_filter}.  This property is inherited to the state component $\YQ$ by the mapping \eqref{transf_Y_QQ}.  Furthermore, the coefficient $\alphaqq$ is quadratic in $\zsmall$. Moreover, we have $\betaq=0$ whereas  $\alphamm$ and $\betam$ are linear in $m$ and $\zsmall$. The jump term $\gammaq$ depends on $ \zsmall$, while the jump term $\gammam$ depends on both $ \zsmall$ and $u$. The derivatives of these functions are bounded in $\zsmall$, and as consequence the inequalities \eqref{lipsch:Delta} and \eqref{growth-delta}  hold. Further, one can easily prove  that { $\overline{\rho}$ can be choosen as 
			$\overline{\rho}(u)=C_u\max\big(\|\varianceexp^{-1/2}\|\|u\|,1\big)$, where $C_u$} is a positive constant.~\qed
	\end{proof}

	Note that the diffusion coefficient  of the regularized  state equation associated to both Brownian motion $\overline{W},W^\ast$ is
	$\big(\widetilde{\beta}^ \top_Y(y),	\frac{1}{\sqrt{\np}}I_{\nZustand} \big)^{\top}$ satisfies the		Lipschitz and growth conditions \eqref{lipsch:b,sigma} and
	\eqref{growth_cond} given in  Assumption \ref{coef-model}, since 	$\widetilde{\beta}_Y(y)$ satisfies these conditions and
	$\frac{1}{\sqrt{\np}}I_{\nZustand}$ does not depend on the state $y$. 
	
	\subsection{\it \textbf{Regularized Control Problem}}
	\label{control_problem_reg}
	
	We now extend the notions of reward function given in \eqref{Zielfkt_H} and
	value function given \eqref{Wertfkt_H} for our original optimization
	problem to the setting of the regularized optimization problem.  We assume that the regularized state process starts at time $t\in[0,T]$ at the initial value $y\in \mathcal{S}^Y$ in the state space of the original control problem. Further, it is controlled by a fixed admissible decision rule $\Pi:[0,T]\times \widetilde{ \mathcal{S}}^Y \to \R^\nAktien$. The corresponding solution to \eqref{filter_pert} is denoted by  ${}^\np\stateprocess_t^{\Pi,t,y}$.				
	The reward and value function of the regularized problems then are defined for $t\in[0,T], y\in \mathcal{S}^Y$ 
	\begin{eqnarray}
		\label{reward_regu} \nonumber {}^\np \reward(t,y;\Pi) &=& \E^{\overline{\P}} 
		\Big[\exp\Big\{\int\nolimits_t^T b(
		{}^\np\stateprocess_s^{\Pi,t,y},\Pi(s,{}^\np\stateprocess_s^{\Pi,t,y})) ds\Big\}\Big]
		\quad \text{for } \Pi\in \mathcal{A},\\
		\nonumber {}^\np V(t,y) & = &\sup_{\Pi \in
			\mathcal{A}}  {}^\np \reward(t,y,\Pi).
	\end{eqnarray}

	\begin{lemma}
		\label{Generator_pi_regular} The generator ${}^\np\generator={}^\np\generator^{p}$ for the regularized state process ${}^\np\stateprocess$
		with dynamics  given in \eqref{filter_pert} is given for a function $f\in C^2(\widetilde{\mathcal S}_{Y})$ by
		\begin{align}
			{}^\np\mathcal{L}^p f(y)&=\diffop_{\filter}^{\top} f(y)\widetilde{\alpha}_M(y,\pointp) +\diffop_{\zsmall}^{\top} f(y)\widetilde{\underline{\alpha}}_Q(y)
			+\frac{1}{2}\trace\Big\{\diffop_{\filter\filter} f(y)\widetilde{\beta}_M(y)\widetilde{\beta}_M(y)^{\top}\Big\}\nonumber\\
			&\quad +\lambda \Big\{\int_{\mathbb R^{\nAktien}} f\big( y+\widetilde{\gamma}_Y(y,u) \big)\varphi(u)du-f(y)  \Big\}
			+\frac{1}{2k}\trace\Big\{\diffop_{yy} f(y)\Big\}.
			\label{new-generator_regular}
		\end{align}
	\end{lemma}
	\begin{proof}
		The proof is analog to proof of  Lemma \ref{Generator_pi}.
		\qed
	\end{proof}
	Now, the associated dynamic programming equation reads as
	\begin{equation}\label{eq:HJB-regularized}
		{}^\np V(t,y)+\sup_{p \in \mathbb
			R^{\nAktien}}\Big\{{}^\np\mathcal{L}^{p}\; {}^\np
		V(t,y)-b(y,p){}^\np V(t,y)\Big\}=0, ~(t,y) \in [0,T) \times
		\widetilde{\mathcal S}_{},
	\end{equation}
	with terminal value ${}^\np V(T,y)=1$. Note that for the generator
	${}^\np\mathcal{L}^{p}$ the ellipticity condition for the
	coefficients of the second derivatives holds since we have for all
	$z\in \mathbb R^{\nZustand}\setminus\{0\}$ and all
	$y\in\widetilde{\mathcal S}_Y$
	\begin{align}\nonumber
		z^{\top}\Big(\widetilde{\beta}_Y(y) \widetilde{\beta}_Y^{\top}(y)+ \frac{1}{2\np} I_{\nZustand}\Big) z=
		z^{\top} \widetilde{\beta}_Y(y) \widetilde{\beta}_Y^{\top}(y) z + \frac{1}{2\np}z^{\top} z=
		\| \widetilde{\beta}_Y(y) z\|^2 ~+\frac{1}{2\np} \|z\|^2 \geq C \|z\|^2
	\end{align}
	where  $C=\frac{1}{2\np}$. Hence the results of Davis and Lleo \cite{Davis and Lleo (2013_2)}
	apply to this dynamic programming equation. According to Theorem 3.8
	of their paper, there is a classical solution ${}^\np V$ for
	\eqref{eq:HJB-regularized}. Moreover, for every $(t,y)$ there exists
	unique maximizer ${}^\np p^*$ of the problem
	$$\sup_{p \in \mathbb R^{\nAktien}} \Big\{{}^\np\mathcal{L}^{p}\; {}^\np V(t,y)-b(y,p) {}^\np V (t,y)\Big\}.$$
	The maximizer ${}^\np p^*$ can be chosen as a Borel-measurable
	function of $t$ and $y$ and the  optimal strategy   ${}^\np \Pi^{\ast}_t = {}^\np
	\Pi^{\ast}(t,y)$ is given similarly as in \eqref{opti_stra_power_C}.



	\subsection{\it \textbf{$\mathcal L_2$-Convergence ${}^\np{\stateprocess} \to \stateprocess $ of State Processes}}\ 
	We now compare the solution ${}^\np \stateprocess$ of the regularized
	state equation \eqref{filter_pert} with the solution $\stateprocess_t$ of the
	unregularized state equation \eqref{State_SDE}
	and study asymptotic properties for $\np\to \infty$. This will be
	crucial for establishing convergence of the associated reward
	function of the regularized problem to the original optimization
	problem. We assume that both processes start at time $t_0\in[0,T]$
	with the same initial value $y\in\widetilde{\mathcal
		S}_Y:=\mathbb R^{ \nZustand}$, i.e,
	${}^\np \stateprocess_{t_0}=\stateprocess_{t_0}=y $  and are controlled by the same admissible decision rule $\Pi\in \mathcal{A}$.  The corresponding solutions are
	denoted by ${}^\np\stateprocess_t^{\Pi,t_0,y}$ and $\stateprocess_t^{\Pi,t_0,y}$. Note that  the regularized state process ${}^\np\stateprocess_t^{\Pi,t_0,y}$  is controlled by the strategy  $\Pi(t,{}^\np\stateprocess_t)$, whereas the strategy for the original state process $\stateprocess_t^{\Pi,t_0,y}$  reads as $\Pi(t,\stateprocess_t)$. 
	\begin{lemma}[Uniform $\mathcal L_2$-convergence w.r.t. $\Pi\in\mathcal{A}$]
		\label{L^2_conv_filter_M}\\				
		Let $t\in[t_0,T]$ and $y\in\mathcal{S}^\stateprocess$ be fixed . Then it		holds
		$$\lim\limits_{\np\to\infty}\E^{\overline{\P}} \Big[ \sup_{t \in [t_0,T]} \|{}^\np\stateprocess_t^{\Pi,t_0,y}-\stateprocess_t^{\Pi,t_0,y}\|^2\Big] =0
		\quad\mbox{uniformly for }\quad \Pi\in \mathcal{A}.$$
	\end{lemma}
	The proof is given in Appendix \ref{proof_conv_state}. 
	\begin{remark}\label{rem_short_notation}
		For a better readability  we  frequently  omit the superscripts ${\Pi,t,y}$ at  $\stateprocess^{\Pi,t,y}$ and ${}^\np\stateprocess^{\Pi,t,y}$ and just write $\stateprocess$ and ${}^\np\stateprocess$, keeping the dependence in mind. Further, we write $\Pi_t, {}^\np\Pi_t$ for $\Pi(t,\stateprocess_t), \Pi(t,{}^\np\stateprocess_s)$.
	\end{remark}
	\subsection{\it \textbf{Auxiliary Results for the Convergence of Reward Functions}}
	\label{aux_conv_reward}

	The above lemma on the convergence of the state processes is a crucial result for the following analysis of the convergence of reward function. We recall their definition in \eqref{Zielfkt_H} for the original problem and  \eqref{reward_regu} for the regularized problem which can be rewritten as 				
	\begin{align}
		\label{eta_def}
		\begin{array}{rll}	
			\reward(t,y;\Pi)&=\E^{\overline{\P}} \big[\exp\{\etaT^{\Pi,t,y}\}\big] &\text{with }\quad \etaT^{\Pi,t,y}= 	\int\limits_t^T b \big( \stateprocess_s^{\Pi,t,y},\Pi(s,\stateprocess_s^{\Pi,t,y})\big) ds,\\[1ex]
			{}^\np\reward(t,y;{}^\np\Pi)&=\E^{\overline{\P}} \big[\exp\{{}^\np \etaT^{\Pi,t,y}\}\big]
			&\text{with }\quad
			{}^\np \etaT^{\Pi,t,y} = \int\limits_t^T b \big({}^\np\stateprocess_s^{{}^\np\Pi,t,y},\Pi(s,{^\np}\stateprocess_s^{\Pi,t,y})\big) ds.
		\end{array}
	\end{align}
	Starting point for the convergence of reward function established below in Theorem \ref{conv_reward} is  Lemma \ref{L_1_Convergence} for which we need the following  estimates.	They are necessary since contrary to the regularization procedure used in  Frey et al. \cite{Frey-Wunderlich-2014}, 
	Shardin and Wunderlich \cite{Shardin and Wunderlich (2017)} in this work the hidden signal is an unbounded Ornstein-Uhlenbeck for which the filter 
	processes take values on the whole	$\mathbb{R}^{\nAktien}$. This creates technical difficulties  which do not allow to directly adopt the approach in \cite{Frey-Wunderlich-2014,Shardin and Wunderlich (2017)} since there 	the drift is driven by a finite-state Markov chain leading to a
	bounded filter given in terms of the conditional probabilities of		the different states.

	\begin{lemma}[Moments of the filter process $\mathbf{\Mpro}$ are bounded]
		\label{Moments_state_pros} \\
		Let $\Mpro=\Mpro^{t,y}\in\mathcal{S}^\stateprocess$ with $y={\filter \choose \zsmall}$ be the filter process of conditional expectation starting at time $t$ with $\Mpro_t^{t,y}=\filter$ and $\Qpro_t=\restri^{-1}(\zsmall)$.
		Then for every  $p\geq  0$ there exists a  constant  $ C_{M,p}>0$ such that the $p$-th-order moment  $\Mpro^{t,y}$ is  bounded on $[t,T]$, i.e, 					
		\begin{align}
			\nonumber  
			\sup_{s\in[t,T]}\E[\|\Mpro_s^{t,y}\|^p]\leq C_{M,p}.
		\end{align}  
	\end{lemma}			
	The proof is given in Appendix \ref{Moments_state_pros_proof}.
	The boundedness of the conditional covariance $\Qpro$, see Lemma  \ref{properties_filter}, is inherited to the second state component $\YQ$ such that there exists a constant $K_G>0$ with $\|\YQ_s\|\le  K_G$ for all $s\in[t,T]$. 				
	Using properties of the maximum norm in $\R^\nZustand$  and shorthand notations from Remark \ref{rem_short_notation} we obtain
	$
	\|\stateprocess_s\| = \max\{\|\Mpro_s\|, \|\YQ_s\|\} \le \max\{\|\Mpro_s\|, K_G\}.$ As a consequence, Lemma \ref{Moments_state_pros} implies  that all moments of original state process ${\mathbf \stateprocess}$ are bounded. 
	
	\begin{cor}[Moments of original state process ${\mathbf \stateprocess}$ are bounded]
		\label{StateY_bounded}	\\
		Let $t\in[0,T]$ and $y\in\mathcal{S}^\stateprocess$ be fixed. Then  	
		for every $p\ge  0$ there	exists a constant $\KBY{p}>0$ such that for the moments of the state
		process $\stateprocess=\stateprocess^{\Pi,t,y}$ satisfying SDE	\eqref{filter_pert} and  starting at time $t$ with $\stateprocess_t^{\Pi,t,y}=y$ it	holds 
		\begin{align}
			\label{StateY_bound} 
			\sup_{s\in[t,T]}		\E^{\overline{\P}}\big[\|\stateprocess_s^{\Pi,t,y}\|^p\big] \le \KBY{p},\quad \text{uniformly   for all $ \Pi\in\mathcal{A}$}.
		\end{align}
	\end{cor}

	\begin{lemma}[Second-order moments of regularized state  processes $\mathbf{{}^\np\stateprocess}$  are bounded]
		\label{StateY_regu_bounded}
		\ \\
		Let $t\in[0,T]$ and $y\in\mathcal{S}^\stateprocess$ be fixed.
		Then there
		exists a constant $\KBYreg{2}$ such that for the second-order moments of the regularized state
		process ${}^\np\stateprocess={}^\np\stateprocess^{\Pi,t,y}$ satisfying SDE	\eqref{State_SDE} and  starting at time $t$ with $\stateprocess_t^{\Pi,t,y}=y$ it	holds      
		\begin{align}
			\label{StateY_bound_regu} 
			\sup_{s\in[t,T]}\sup_{\np}	\E^{\overline{\P}}\big[\|{}^\np\stateprocess_s^{\Pi,t,y}\|^2\big] \le\KBYreg{2},
			\quad \text{uniformly   for all $ \Pi\in\mathcal{A}$}.
		\end{align}				
	\end{lemma}
	
	\begin{proof}
		Using the shorthand notations from Remark \ref{rem_short_notation}, applying  Corollary  \ref{StateY_bounded} and 		the  $\mathcal L_2$-convergence ${}^\np{\stateprocess}_s \to \stateprocess_s $	  it holds for all $s\in[t,T],\; k\in\N$, and uniformly   for all $ \Pi\in\mathcal{A}$
		$$
		\E^{\overline{\P}}\big[\|{}^\np\stateprocess_s\|^2\big]  = \E^{\overline{\P}}\big[\|{}^\np\stateprocess_s + \stateprocess_s-\stateprocess_s\|^2\big] \le 2 	\E^{\overline{\P}}\big[\|\stateprocess_s\|^2\big]	+ 2 \E^{\overline{\P}}\big[\|{}^\np\stateprocess_s - \stateprocess_s\|^2\big]	 	\le 2C_{Y,2} + C: = \KBYreg{2}
		$$
		for some $C>0$.
		\qed
	\end{proof}

	\begin{lemma}[Uniform $\mathcal L_1$-convergence ${{}^\np\etaT^{\Pi,t,y} \to \etaT^{\Pi,t,y}}	$ for all ${\Pi\in\mathcal{A}}$]
		\label{L_1_Convergence}\\	
		Let $t\in[0,T]$ and $y\in\mathcal{S}^\stateprocess$ be fixed.				
		Then it			holds 
		$$\lim\limits_{\np\to\infty}\E^{\overline{\P}} \Big[  \big|{}^\np\etaT^{\Pi,t,y}-\etaT^{\Pi,t,y}\big|\Big] =0
		\quad\text{uniformly for all}\quad \Pi\in \mathcal{A}.$$				
	\end{lemma}
	\begin{proof}
		Using the shorthand notations from Remark \ref{rem_short_notation} it holds  for every $t\in [0,T]$
		\begin{align}
			\E^{\overline{\P}}  \big[\big{|} {}^\np\etaT^{\Pi,t,y}-\etaT^{\Pi,t,y}
			\big{|}\Big]
			&= \E^{\overline{\P}} \Big[\Big|\int_t^T\big( b ({}^\np\stateprocess_s^{\Pi,t,y},{}^\np\Pi_s)-b ( \stateprocess_s^{\Pi,t,y},\Pi_s)\big) ds \big|\Big]\nonumber\\
			&\leq\E^{\overline{\P}} \Big[ \int_t^T\Big|  b ({}^\np\stateprocess_s,{}^\np\Pi_s)-b ( \stateprocess_s,\Pi_s) \Big|\; ds\Big].
			\label{estimate_a}
		\end{align}
		Recall that 
		$b(y,p)=\theta\big(\pointp^{\top}\filter-\frac{1-\theta}{2}\|\sigma_{\wealth}\pointp\|^2\big)$ for $y={\filter \choose q}$ and set $C_{\theta}=\frac{1}{2}\theta(1-\theta)$. Then, for the integrand in the last inequality we have
		\begin{align}
			\hspace{-1.5cm}\Big|  b ({}^\np\stateprocess_s,{}^\np\Pi_s)-b ( \stateprocess_s,\Pi_s) \Big|
			\leq \theta|\Pi_s^{\top}\Mpro_s-{}^\np\Pi_s^{\top}{}^\np\Mpro_s|+C_{\theta}\big|\|\sigma_{\wealth}\Pi\|^2-\|\sigma_{\wealth}{}^\np\Pi\|^2\big|.\label{estimation_b}
		\end{align}
		\medskip	\paragraph{Estimation of the first term on the right-hand side of \eqref{estimation_b}}
		It holds
		\begin{align}
			\big|\Pi_s^{\top}\Mpro_s-{}^\np\Pi_s^{\top}{}^\np\Mpro_s|
			&= |\Pi_s^{\top}\Mpro_s-\Pi_s^{\top}{}^\np\Mpro_s+\Pi_s^{\top}{}^\np\Mpro_s-{}^\np\Pi_s^{\top}{}^\np\Mpro_s\big|\nonumber\\
			&\leq \big|\Pi_s^{\top}(\Mpro_s-{}^\np\Mpro_s)\big|+\big|(\Pi_s-{}^\np\Pi_s)^{\top}{}^\np\Mpro_s\big|\nonumber\\
			&\leq\big|\Pi_s^{\top}(\Mpro_s-{}^\np\Mpro_s)\big|+\|(\Pi_s-{}^\np\Pi_s)\|\|{}^\np\Mpro_s\|\nonumber\\
			&\leq\big|\Pi_s^{\top}(\Mpro_s-{}^\np\Mpro_s)\big|+C_L\|(\stateprocess_s-{}^\np\stateprocess_s)\|\|{}^\np\Mpro_s\|.
			\label{estimation_c}
		\end{align}
		Integrating, taking expectation  it follows  from \eqref{Linear-growth-rule} and Corollary \ref{StateY_bounded} for the first term 
		\begin{align}
			\E^{\overline{\P}} \Big[ \int_t^T\big|\Pi_s^{\top}(\Mpro_s-{}^\np\Mpro_s)\big|\; ds\Big]
			& \le  \int_t^T \E^{\overline{\P}} \big[  \|\Pi_s \| ~ \|{}^\np\Mpro_s-\Mpro_s \|\big] ds\\
			&  \le  \int_t^T \big(\E^{\overline{\P}} \big[ \|\Pi_s \|^2\big] \big)^{1/2} \big(\E^{\overline{\P}} \big[\|{}^\np\Mpro_s-\Mpro_s \|^2 \big] \big)^{1/2}ds\\				
			&  \le  \int_t^T \big(\E^{\overline{\P}} \big[ C_G^2(1+\|Y_s\|)^2\big] \big)^{1/2} \big(\sup_{ s\in [t,T]}\E^{\overline{\P}} \big[\|{}^\np\Mpro_s-\Mpro_s \|^2 \big] \big)^{1/2}ds\\	
			& \le  \int_t^T C_G\big(2+2C_{Y,2} \big)^{1/2} \big(\sup_{ s\in [t,T]}\E^{\overline{\P}} \big[\|{}^\np\stateprocess_s-\stateprocess_s \|^2 \big] \big)^{1/2}ds\\						
			&  \le  C_G\big(2+2C_{Y,2} \big)^{1/2} (T-t)\big(\sup_{ s\in [t,T]}\E^{\overline{\P}} \big[\|{}^\np\stateprocess_s-\stateprocess_s \|^2 \big] \big)^{1/2}.				
		\end{align}		
		Lemma \ref{L^2_conv_filter_M} implies that 
		the last term converges to $0$ as  $\np\to\infty$ and uniformly w.r.t $\Pi\in\mathcal{A}$. 
		For  the second term in \eqref{estimation_c} we apply Cauchy-Schwarz Inequality to obtain
		\begin{align}
			\E^{\overline{\P}} \Big[ \int_t^T\|(\stateprocess_s-{}^\np\stateprocess_s)\|\|{}^\np\Mpro_s\|\; ds\Big]
			\leq \int_t^T\big(\E^{\overline{\P}}  \|{}^\np\stateprocess_s-\stateprocess_s \|^2 \big)^{1/2}\big(\E^{\overline{\P}}  \|{}^\np\Mpro_s\|^2\big)^{1/2}ds\nonumber\\
			\leq\;(T-t) \Big(\sup_{ s\in [t,T]}\E^{\overline{\P}} \big[\|{}^\np \stateprocess_s- \stateprocess_s\|^2\big]\Big)^{1/2} \KBYreg{2}
		\end{align}
		where we applied Lemma \ref{StateY_regu_bounded} on the boundedness of moments of ${}^\np\stateprocess$. From Lemma \ref{L^2_conv_filter_M}   it follows that 
		the last term converges to $0$ as  $\np\to\infty$ and uniformly w.r.t $\Pi\in\mathcal{A}$.\\ 
		\paragraph{Estimation of the second term on the right-hand side of \eqref{estimation_b}} It holds 
		\begin{align}
			\big|\|\sigma_{\wealth}\Pi_s\|^2-\|\sigma_{\wealth}{}^\np\Pi_s\|^2\big|
			&= \big|\|\sigma_{\wealth}\Pi_s\|-\|\sigma_{\wealth}{}^\np\Pi_s\|\big|~
			\big|\|\sigma_{\wealth}\Pi\|+\|\sigma_{\wealth}{}^\np\Pi\|\big|
			\nonumber\\
			&\leq\|\sigma_{\wealth}(\Pi_s-{}^\np\Pi_s)\| \cdot\big(\|\sigma_{\wealth}\Pi\|+\|\sigma_{\wealth}{}^\np\Pi\|\big)\nonumber\\
			&\leq\|\sigma_{\wealth}\|^2\|(\Pi_s-{}^\np\Pi_s)\|~~~C_G\big((1+\|\stateprocess_s\|)+(1+\|{}^\np\stateprocess_s\|)\big)\nonumber\\
			&\leq \|\sigma_{\wealth}\|^2C_M\|(\stateprocess_s-{}^\np\stateprocess_s\| ~~~C_G\big(2+\|\stateprocess_s\|+\|{}^\np\stateprocess_s\|\big)\nonumber\\
			&\leq C_GC_M \|\sigma_{\wealth}\|^2\big(2+\|\stateprocess_s\|+\sup_{\np}\|{}^\np\stateprocess_s\|\big)\|(\stateprocess_s-{}^\np\stateprocess_s)\|.\nonumber
		\end{align}				
		By Corollary \ref{StateY_bounded} and Lemma \ref{StateY_regu_bounded} we have
		$\E^{\overline{\P}}\|\stateprocess_s\|^2<\infty$~and $\sup_{\np}\E^{\overline{\P}}\|{}^\np\stateprocess_s\|^2<\infty$ so that  
		by  taking  expectation and using Cauchy-Schwarz Inequality, it holds 				
		\begin{align}
			\E^{\overline{\P}}\big[\big(2+\|\stateprocess_s\|+\sup_{\np}\|{}^\np\stateprocess_s\|\big)\|(\stateprocess_s-{}^\np\stateprocess_s)\|\big] 
			& \le   \sqrt{3}\Big(4+\E^{\overline{\P}}[\|\stateprocess_s\|^2]+\sup_{\np}\E^{\overline{\P}}[\|{}^\np\stateprocess_s\|^2]\Big)^{\frac{1}{2}} \times \\
			& \big(\E^{\overline{\P}}\big[\|(\stateprocess_s-{}^\np\stateprocess_s)\|^2]\big)^{\frac{1}{2}}
			\leq C \big( \E^{\overline{\P}}\big[\|(\stateprocess_s-{}^\np\stateprocess_s)\|^2]\big)^{\frac{1}{2}}, 
		\end{align}
		for some  constant $C>0$.  Integrating yields 
		
		\begin{align}
			\int_t^T  \E^{\overline{\P}} \big[\big|\|\sigma_{\wealth}\Pi_s\|^2-\|\sigma_{\wealth}{}^\np\Pi_s\|^2\big|\big]\; ds
			\leq\;(T-t) C \big(\sup_{s\in [t,T]}\E^{\overline{\P}} \big[\|{}^\np \stateprocess_s- \stateprocess_s\|^2\big]\big)^{1/2}
		\end{align}
		which according to Lemma \ref{L^2_conv_filter_M} converges to zero for $k\to\infty$. Finally, we 
		substitute   the derived estimations into \eqref{estimation_b} and \eqref{estimate_a} which proves the claim.~\qed
	\end{proof}

	\subsection{\it \textbf{ Convergence of  Reward  Functions}}
	\label{subsec_conv_reward}
	In this subsection we present another main result consisting in the uniform convergence of  the reward functions. For this fact we rely on the concept of uniform integrability  which is of a basic importance in connection with convergence of moments. To allow for the uniform integrability we need the following assumption 
	\begin{assumption}\label{main_assumption}\  \\
		Let $t\in[0,T]$ and $y\in\mathcal{S}^\stateprocess$ be fixed.	
		Then there there exist constants  $\delta, C_\delta>0$  such that for the random variables  ${}^\np\etaT^{\Pi,t,y}$  defined in \eqref{eta_def}
		it holds 
		\begin{align}
			{\sup_{k\in\N}}\;\E^{\overline{\P}}\big[\exp\big\{(1+\delta)\; {}^\np\etaT^{\Pi,t,y}\big\}\big]\leq C_{\delta},\quad \text{uniformly for all}\quad \Pi\in \mathcal{A}.
			\label{eta_introduce}
		\end{align} 
	\end{assumption}
	The next theorem on the uniform convergence of reward functions is
	our main result, convergence of the value function and
	$\varepsilon$-optimality of ${}^\np \Pi$ follows easily
	from this theorem.
	\begin{theorem}[Uniform convergence of reward functions]
		\label{conv_reward}\\
		Under Assumptions  \ref{admi_stra_rule} and  \ref{main_assumption} it
		holds for   $t\in[0,T],~y\in \mathcal S_{Y}$
		\[ \sup\limits_{\Pi\in\mathcal{A}} |{}^\np \reward(t,y;\Pi) - \reward(t,y;\Pi)|\to 0 \quad\text{for } \np\to\infty.\]
	\end{theorem}
	\begin{proof}
		From \eqref{eta_def} We have
		$$\reward(t,y;\Pi)=\E^{\overline{\P}} \big[\exp\{\etaT^{\Pi,t,y}\}\big]\quad \text{and}\quad {}^\np
		\reward(t,y;\Pi)=\E^{\overline{\P}} \big[\exp\{{}^\np \etaT^{\Pi,t,y}\}\big].$$
		Using shorthand notations from Remark \ref{rem_short_notation} the uniform convergence of 
		reward functions reads as a uniform convergence of exponential  moments,  i.e, 
		$$
		\big|\E^{\overline{\P}}\big[\exp\{{}^\np \etaT\}\big]-\E^{\overline{\P}}\big[\exp\{ \etaT\}\big]\big|\to 0 \quad\text{for } \np\to\infty, \quad t\in[0,T],~y\in \mathcal S_{Y}.
		$$
		According to Chung~\cite[Theo.~4.5.4]{Chung (1968)} the above convergence is ensured if the sequence of  random variables
		$(\exp\{{}^\np \etaT\})_k$ is uniformly integrable and converges in distribution to 
		$\exp\{\etaT\}$. The uniform $\mathcal L_1$-convergence ${}^\np\eta_s^{\Pi,t,y} \to \eta_s^{\Pi,t,y} $ for all $\Pi\in\mathcal{A}$ given in Lemma \ref{L_1_Convergence} implies the uniform convergence in distribution ${}^\np\etaT^{\Pi,t,y} \to \etaT^{\Pi,t,y} $, and accordingly it yields the uniform convergence in distribution  
		$\exp\{{}^\np \etaT\} \to \exp\{ \etaT\} $for all $\Pi\in\mathcal{A}$ since the function $x\to e^x$ is continuous. To conclude the proof we have to show that the sequence of the random variables
		$(\exp\{{}^\np \etaT\})_k$ is uniformly integrable. i.e.,
		\[
		\lim_{\alpha\to\infty}\sup_k\int_{\{\exp\{{}^\np \etaT\}\geq\alpha\}} \exp\{{}^\np \etaT\}d\overline{\P}=0.
		\]
		In fact, we set  $X_k=\exp\{{}^\np \etaT\}$ and $X=\exp\{ \etaT\}$ and note that $X_k, X>0$.  Assumption \ref{main_assumption} implies $\sup_k\E^{\overline{\P}}[ X_k^{1+\delta}] \le C_\delta$ for some $\delta, C_\delta>0$.  Using Hölder's Inequality yields
		\begin{align}
			\int_{\{X_k\geq\alpha\}} X_kd\overline{\P}		= \E^{\overline{\P}}\big[ \one_{\{X_k\geq\alpha\}}X_k\big]
			\leq\Big(\E^{\overline{\P}}\big[ X_k^{1+\delta}\big]\Big)^{\frac{1}{1+\delta}}
			\Big(\E^{\overline{\P}}\big[ \one_{\{X_k\geq\alpha\}}\big]\Big)^{\frac{\delta}{1+\delta}}.\label{inequa1}
		\end{align}
		Markov's Inequality implies 
		$$
		\E^{\overline{\P}}\big[ \one_{\{X_k\geq\alpha\}}\big]=\overline{\P}(X_k\geq\alpha)=\overline{\P}\big(X_k^{1+\delta}\geq\alpha^{1+\delta}\big)\le \frac{1}{\alpha^{1+\delta}}\E^{\overline{\P}}[X_k^{1+\delta}].
		$$
		Substituting into \eqref{inequa1} yields
		\begin{align}
			\int_{\{X_k\geq\alpha\}} X_kd\overline{\P} &\leq\Big(\E^{\overline{\P}}\big[ X_k^{1+\delta}\big]\Big)^{\frac{1}{1+\delta}}~
			\frac{1}{\alpha^{\delta}} \Big(\E^{\overline{\P}}\big[ X_k^{1+\delta}\big]\Big)^{\frac{\delta}{1+\delta}}\\
			&\leq \frac{1}{\alpha^{\delta}}\E^{\overline{\P}}\big[ X_k^{1+\delta}\big]\leq \frac{1}{\alpha^{\delta}}\sup_k\E^{\overline{\P}}\big[ X_k^{1+\delta}\big],
		\end{align}
		which allows to conclude that
		\[
		\lim_{\alpha\to\infty}\sup_k\int_{\{\exp\{{}^\np \etaT\}\geq\alpha\}} \exp\{{}^\np \etaT\}d\overline{\P}=0
		\]
		since $\sup_k\E^{\overline{\P}}\big[ X_k^{1+\delta}\big]\le C_{\delta}$ by Assumption \ref{main_assumption}.
		\qed
	\end{proof}
	\subsection{\it \textbf{Convergence of Value Functions}}
	\label{subsec_conv_value}
	Finally we show  the convergence of value functions and  that the optimal  decision rule ${}^\np \Pi^{\ast}$ for
	the regularized problem \eqref{eq:HJB-regularized} is
	$\varepsilon$-optimal in the original problem. This gives a method
	for computing (nearly) optimal strategies.   The results and proofs are similar to Frey et al.~\cite{Frey-Wunderlich-2014}. For the convenience of the reader we
	give the proofs of the next two  of corollaries in Appendix \ref{proof_conv_value} and \ref{proof_eps_opt}.
	\begin{cor}[Convergence of value functions]
		\label{conv_value}		\ \\	
		Under Assumptions  \ref{admi_stra_rule} and  \ref{main_assumption} it	holds for   $t\in[0,T],~y\in \mathcal S_{Y}$
		\[{}^\np V(t,y) \to V(t,y) \quad\text{for } \np\to\infty, \quad t\in[0,T], ~y\in \mathcal S_Y.\]
	\end{cor}

	\begin{cor}[$\varepsilon$-Optimality]
		\label{eps_opt}  \\ 
		Under Assumptions  \ref{admi_stra_rule} and  \ref{main_assumption} it	holds for   $t\in[0,T],~y\in \mathcal S_{Y}$ that  for every  $\varepsilon>0$ there exists  $\np_0\in\mathbb N$,
		so that
		\[|V(t,y)- \reward(t,y;{}^\np \Pi^{\ast})| \leq\varepsilon\quad\text{for }\np\geq \np_0,\]
		meaning that  for $\np\ge\np_0$ the optimal decision rules  ${}^\np \Pi^{\ast}$ of the regularized problems are $\varepsilon$-optimal   for the original control problem.
	\end{cor}

	\begin{appendix}					
		\section*{Appendix }
		\section{Proof of Lemma \ref{L^2_conv_filter_M} }		
		\label{proof_conv_state}	
		\begin{proof}
			Without loss of generality let $t_0=0$. Using the integral notation of  the
			regularized and non-regularized state equation  and shorthand notations from Remark \ref{rem_short_notation} we get
			$${}^\np\stateprocess_t - \stateprocess_t = {}^\np \Phi_t +{}^\np \Psi_t$$
			with
			\begin{align}
				{}^\np \Phi_t& := \int\nolimits_0^t \big[\widetilde{\alpha}_{\stateprocess}\big({}^\np\stateprocess_s,\Pi(s,{}^\np\stateprocess_s)\big) - \widetilde{\alpha}_{\stateprocess}\big(\stateprocess_s,\Pi(s,\stateprocess_s)\big)\big] ds\quad \text{and} \nonumber\\
				{}^\np \Psi_t & =  \int\nolimits_0^t \big(\widetilde{\beta}_{\stateprocess}({}^\np\stateprocess_s)-\widetilde{\beta}_{\stateprocess}(\stateprocess_s)\big)  d\overline{W}_s + \int\nolimits_0^t \int\nolimits_{\mathbb R^{\nAktien}} \big(\widetilde{\gamma}_{\stateprocess}({}^\np\stateprocess_s,u)-\widetilde{\gamma}_{\stateprocess}(\stateprocess_s,u)\big)\komppoi(ds, du)
				+ \frac{1}{\sqrt{\np}} \,d W_t^{\ast} \nonumber.
			\end{align}
			We now set $~{}^\np D_t:=\E^{\overline{\P}}\Big[ \sup\limits_{s\leq t} \big\|{}^\np\stateprocess_s-\stateprocess_s \big\|^2\Big]$, then it holds
			\begin{align}
				{}^\np D_t &= \E^{\overline{\P}}\Big[ \sup_{s\leq t}  \big\| {}^\np \Phi_s + {}^\np \Psi_s  \big\|^2\Big]
				\leq2 \E^{\overline{\P}}\Big[ \sup_{s\leq t}  \big\|{}^\np \Phi_s  \big\|^2\Big] + 2 \E^{\overline{\P}}\Big[ \sup_{s\leq t}  \big\| {}^\np \Psi_s  \big\|^2\Big].
				\label{Gt_def_M}
			\end{align}
			\paragraph{Estimation of ${}^\np \Phi$}\\ 
			Using Cauchy-Schwartz inequality we obtain for the first term on the right-hand side 
			\begin{align}
				\sup_{s\leq t} \big\| {}^\np \Phi_s  \big\|^2& =  \sup_{s\leq t} \Big\| \int\nolimits_0^s \Big[\widetilde{\alpha}_{\stateprocess}\big({}^\np\stateprocess_u,\Pi(u,{}^\np\stateprocess_u)\big) - \widetilde{\alpha}_{\stateprocess}\big(\stateprocess_u,\Pi(u,\stateprocess_u)\big)\big] du \Big\|^2\nonumber\\
				&\le \sup_{s\leq t} \; s\cdot \int\nolimits_0^s \big\| \widetilde{\alpha}_{\stateprocess}\big({}^\np\stateprocess_u,\Pi(u,{}^\np\stateprocess_u)\big) - \widetilde{\alpha}_{\stateprocess}\big(\stateprocess_u,\Pi(u,\stateprocess_u)\big) \big\|^2 du\nonumber\\
				&\le t\cdot \int\nolimits_0^t \big\| \widetilde{\alpha}_{\stateprocess}\big({}^\np\stateprocess_u,\Pi(u,{}^\np\stateprocess_u)\big) - \widetilde{\alpha}_{\stateprocess}\big(\stateprocess_u,\Pi(u,\stateprocess_u)\big) \big\|^2 du.
				\label{Reg_Y_cond_M}
			\end{align}					
			Recall that    for $y={m \choose \zsmall}$ and $q=\restri^{-1}(g)$ we have  $\widetilde{\alpha}_{\stateprocess}(y,p)= {\widetilde{\alpha}_M( m, q,p) \choose \widetilde{\alpha}_G(\zsmall)}$.  Using properties of the maximum-norm   we obtain					
			$$
			\big\| \widetilde{\alpha}_{\stateprocess}\big({}^\np\stateprocess_u,\Pi(u,{}^\np\stateprocess_u)\big) - \widetilde{\alpha}_{\stateprocess}\big(\stateprocess_u,\Pi(u,\stateprocess_u)\big)  \big\|^2=\max\{X_M,~X_\YQ\}
			$$
			with 
			$	X_M= \big\| \widetilde{\alpha}_M({}^\np \Mpro_u,{}^\np \Qpro_u,{}^\np\Pi_u) - \widetilde{\alpha}_M(\Mpro_u,\Qpro_u,\Pi_u)  \big\|^2$ \text{and} $X_\YQ= \big\| \widetilde{\alpha}_\YQ({}^\np \YQ_u) - \widetilde{\alpha}_\YQ(\YQ_u)  \big\|^2	$ 
			where we write as announced in Remark \ref{rem_short_notation}  ${}^\np\Pi_u$ for $\Pi(u,{}^\np\stateprocess_u)$  and $\Pi_u$ for $\Pi(u,\stateprocess_u)$.\\
			\smallskip	
			\paragraph{Estimation of $X_M$}\\ 
			Recall  $\widetilde{\alpha}_M(m,q,p)=\big(\revspeed(\revlevel-\filter)+\theta \variance p\big)d_{\varepsilon}$ with $d_{\varepsilon}=\big(1-\frac{1}{\eps}\dist(\zsmall, \mathcal S^{\YQ}) \big) \wedge 0 \in  [0,1]$. Then, it holds					
			\begin{align}
				\hspace{-0.2cm}X_M
				&=\big\| \widetilde{\alpha}_M({}^\np \Mpro_u,{}^\np \Qpro_u,{}^\np\Pi_u) - \widetilde{\alpha}_M(\Mpro_u,\Qpro_u,{}^\np\Pi_u)+
				\widetilde{\alpha}_M(\Mpro_u,\Qpro_u,{}^\np\Pi_u)
				-\widetilde{\alpha}_M(\Mpro_u,\Qpro_u,\Pi_u)  \big\|^2\\
				&\leq ~~~2\big\| \widetilde{\alpha}_M({}^\np \Mpro_u,{}^\np \Qpro_u,{}^\np\Pi_u) - \widetilde{\alpha}_M(\Mpro_u,\Qpro_u,{}^\np\Pi_u)\big\|^2+2\big\|	\widetilde{\alpha}_M(\Mpro_u,\Qpro_u,{}^\np\Pi_u)	-\widetilde{\alpha}_M(\Mpro_u,\Qpro_u,\Pi_u)  \big\|^2\\\nonumber	
				& \leq 2\widetilde{C}_M^2 \big\| {}^\np\stateprocess_u -\stateprocess_u   \big\|^2+2\theta^2\;d_{\varepsilon}^2  \big\|\Qpro_u {}^\np\Pi_u-\Qpro_u\Pi_u\big\|^2\\
				& \leq 2\widetilde{C}_M^2 \big\| {}^\np\stateprocess_u -\stateprocess_u   \big\|^2+2\theta^2\;d_{\varepsilon}^2  \|\Qpro_u\|^2\big\|{}^\np\Pi_u-\Pi_u  \big\|^2\\
				&\leq  2\widetilde{C}_M^2 \big\| {}^\np\stateprocess_u -\stateprocess_u   \big\|^2+2\theta^2\;C_M^2  C_Q^2 \big\| {}^\np\stateprocess_u -\stateprocess_u   \big\|^2.
			\end{align}	
			Note that to obtain the last inequality we have used  the Lipschitz condition \eqref{lipsch:b,sigma} for the term $\widetilde{\alpha}_M$, the Lipschitz condition \eqref{Lipschitz_rule} for  $\Pi$ and Lemma \ref{properties_filter} for the boundedness of $\Qpro$.\\	
			\paragraph{Estimation of $X_G$} \\ 
			Applying  the Lipschitz condition \eqref{lipsch:b,sigma} to  $\widetilde{\alpha}_{\Qpro}$ we obtain
			\begin{align}
				X_Q&= \big\| \widetilde \alpha_{\YQ}({}^\np \YQ_u) - \widetilde{\alpha}_{\YQ}(\YQ_u)  \big\|^2 \leq {\widetilde{C}_M^2} \big\| {}^\np\YQ_u -\YQ_u   \big\|^2\leq {\widetilde{C}_M^2} \big\| {}^\np\stateprocess_u -\stateprocess_u   \big\|^2.
				\nonumber
			\end{align}
			Substituting  into \eqref{Reg_Y_cond_M} yields 
			\begin{align}
				\hspace{-0.8cm}\E^{\overline{\P}}\Big[\sup_{s\leq t} \big\| {}^\np \Phi_s  \big\|^2\Big]
				&\leq C \!\int_0^t  \E^{\overline{\P}}\big[   \big\|  {}^\np\stateprocess_u -\stateprocess_u   \big\|^2\big]\; du \leq C\! \int_0^t  \E^{\overline{\P}}\big[\sup_{v\leq u} \big\| {}^\np\stateprocess_v - \stateprocess_v \big\|^2\big]\; du
				= C\int_0^t {}^\np D_u\; du,
				\label{Reg_Y_cond_2_M}
			\end{align}
			for some  constant $C>0$.\\
			\paragraph{Estimation of $^{\np}\Psi$} \\ Applying Doob's Inequality for martingales to the second term on the right-hand side of 
			\eqref{Gt_def_M} yields the following
			\begin{align}
				\E^{\overline{\P}}\Big[\sup_{s\leq t} \big\| {}^\np \Psi_s  \big\|^2\Big] &\leq 4 \E^{\overline{\P}}\Big[  \big\| {}^\np \Psi_t  \big\|^2\Big]\nonumber\\
				&= 4  \Big( \int\nolimits_0^t \E^{\overline{\P}}\Big[\operatorname{tr}\Big\{ \big(\widetilde{\beta}_{\stateprocess}({}^\np\stateprocess_s)-\widetilde{\beta}_{\stateprocess}(\stateprocess_s)\Big)^{\top} \Big(\widetilde{\beta}_{\stateprocess}({}^\np\stateprocess_s)-\widetilde{\beta}_{\stateprocess}(\stateprocess_s)\Big)\Big\}\Big]  ds\nonumber \\
				& ~~~~+   \int\nolimits_0^t \int\nolimits_{\mathbb R^{\nAktien}} \E^{\overline{\P}}\Big[ \big\|\widetilde{\gamma}_{\stateprocess}({}^\np\stateprocess_s,u)-\widetilde{\gamma}_{\stateprocess}(\stateprocess_s,u))  \big\|^2\Big] \varphi(u) du ds + \frac{\nZustand t}{\np}\Big). 
				\label{M_esti_M}
			\end{align}
			Using the  Lipschitz condition \eqref{lipsch:b,sigma} for the coefficient $\widetilde{\beta}_{\stateprocess}$ it holds
			\begin{align}
				\E^{\overline{\P}}\Big[\operatorname{tr}\Big\{ \Big(\widetilde{\beta}_{\stateprocess}({}^\np\stateprocess_s)-\widetilde{\beta}_{\stateprocess}(\stateprocess_s)\Big)^{\top} \Big(\widetilde{\beta}_{\stateprocess}({}^\np\stateprocess_s)-\widetilde{\beta}_{\stateprocess}(\stateprocess_s)\Big)\Big\}\Big]
				\leq {\widetilde{C}_M^2} \;\E^{\overline{\P}} \Big[ \big\| {}^\np\stateprocess_s - \stateprocess_s  \big\|^2\Big]\nonumber\\
				\leq {\widetilde{C}_M^2}\; \E^{\overline{\P}}\big[\sup_{v\leq s}  \big\| {}^\np\stateprocess_v - \stateprocess_v  \big\|^2\big]=  {\widetilde{C}_M^2}\; {}^\np D_s,
				\label{rs_1_M}
			\end{align}
			whereas  the  Lipschitz condition \eqref{lipsch:Delta} for the coefficient $\widetilde{\gamma}_{\stateprocess}$ yields
			\begin{align}
				\E^{\overline{\P}}\Big[ \big\| \widetilde{\gamma}_{\stateprocess}({}^\np\stateprocess_s,u)-\widetilde{\gamma}_{\stateprocess}(\stateprocess_s,u)  \big\|^2\Big]
				&\leq \overline{\rho}^2(u)\;\E^{\overline{\P}} \Big[ \big\| {}^\np\stateprocess_s - \stateprocess_s  \big\|^2\Big]\nonumber\\
				&\leq \overline{\rho}^2(u)\; \E^{\overline{\P}}\Big[\sup_{v\leq s}  \big\|{}^\np\stateprocess_v - \stateprocess_v  \big\|^2\Big]
				\leq \overline{\rho}^2(u)\; {}^\np D_s.
				\label{rs_2_M}
			\end{align}
			We now substitute  \eqref{rs_1_M} and \eqref{rs_2_M} into \eqref{M_esti_M} so that we obtain for a generic constant  $C>0$
			\begin{align}
				\E^{\overline{\P}}\Big[ \sup_{s\leq t} \big\| {}^\np \Psi_s  \big\|^2\Big] 
				&\leq 4  \Big(\int\nolimits_0^t{\widetilde{C}_M^2}~ {}^\np D_s ds + \int\nolimits_0^t {}^\np D_s  ds\int\nolimits_{\mathbb R^{\nAktien}} \overline{\rho}^2(u) \varphi(u)du+ \frac{ \nZustand t}{\np}\Big)
				\nonumber\\&
				\leq C \int\nolimits_0^t  {}^\np D_s ds +\frac{4  \nZustand t}{\np}. \label{M_esti2_M}
			\end{align}
			\paragraph{Conclusion} \quad Substituting  \eqref{Reg_Y_cond_2_M} and \eqref{M_esti2_M} into \eqref{Gt_def_M} yields 
			${}^\np D_t \leq\frac{4  \nZustand T}{\np} + C \int\nolimits_0^t {}^\np D_s ds .$
			Finally, we apply Gronwall's Lemma to end up with the following estimate
			\[{}^\np D_T \leq\frac{4  \nZustand T}{\np}\; e^{CT} \to 0 \quad\text{f{\"u}r } \np\to \infty,\]
			which proves the claim.
			\qed
		\end{proof}
		\section{Proof of Lemma \ref{Moments_state_pros} }
		\label{Moments_state_pros_proof} 
		\begin{proof}						
			We apply    \cite[Lemma 2.1]{Rosinski and Suchanecki (1980)} stating that  for a Gaussian random vector  $U\in\R^n$ and for all $p,q>0 $  there exists a constant $C_{p,q}>0$ such that 
			\begin{align}
				\label{Moments_Gauss_pros}  
				\E^{\overline{\P}}[\|U\|^p]\leq C_{p,q}\big(\E^{\overline{\P}}[\|U\|^q]\big)^{p/q}.
			\end{align}
			Choosing  $q=2$ and  $C_p=C_{p,2}$ it yields
			\begin{align}
				\label{Moments_Gauss2}  
				\E^{\overline{\P}}[\|U\|^p]\leq C_p\big(\E^{\overline{\P}}[\|U\|^2]\big)^{p/2}.
			\end{align}		
			
			Without loss of generality let $t=0$. 
			The filter process $\Mpro=\Mpro^{0,y}$  given in Lemma \ref{filter_C} by the SDE \eqref{filter_C0} and the update formula \eqref{filter_C_update} enjoys  the following closed-form solution  given in terms of the  recursion for $k=0,1,\ldots$, starting with $\Mpro_0=m,   \Qpro_0=\restri^{-1}(g)$, 
			\begin{align}
				\Mpro_s &= \revlevel + e^{-\revspeed (s-T_k)}(M_{T_k}-\revlevel) 	+  \int_{T_k}^s e^{-\revspeed (s-u)} \Qpro_u\,
				\Sigma_R^{-1/2}d\widetilde{W}_u \quad \text{for } s\in [T_k,T_{k+1})\\
				\Mpro_{T_{k+1}}& =\rho_k\Mpro_{T_{k+1}-}+(I_d-\rho_k)Z_{k+1},
			\end{align}
			and  $\Qpro$ as defined in given in Lemma \ref{filter_C}.
			The expression in the first line can  be derived  by 
			computing the differential  which shows that $\Mpro$ satisfies SDE \eqref{filter_C0}. It shows that conditioned on the arrival times $T_1,T_2,\ldots$ of the expert opinions, the process $\Mpro$ is Gaussian. This follows from the fact that given the arrival dates   the conditional covariance   $\Qpro$ is deterministic and bounded. Thus  at the arrival dates we have Gaussian updates and between two arrival dates $T_k$ and $T_{k+1}$ the process $\Mpro$ is Gaussian and  $\Mpro_s$ is a Gaussian vector with $\mathcal{F}_{T_k}^\HC$-conditional  mean $\revlevel + e^{-\revspeed (t-T_k)}(M_{T_k}-\revlevel ) $ and  covariance matrix  $\int_{T_k}^s e^{-2\revspeed (s-u)}  \Qpro_u\,\Sigma_R^{-1}\Qpro_u du$.
			
			Let us  introduce the $\sigma$-algebras $\AlgT_s=\sigma\{T_k, T_k\le s\}$ generated by the arrival times up to time $s\in[0,T]$. 	
			Applying  \eqref{Moments_Gauss2} and that conditioned on $\AlgT_s$ the random variable  $\Mpro_s$ is Gaussian    we obtain
			\begin{align}
				\label{Moments_Gauss_h}
				\E^{\overline{\P}}[\|\Mpro_s\|^p]=\E^{\overline{\P}}\big[\E^{\overline{\P}}[\,\|\Mpro_s\|^p|\AlgT_s]\,\big]\leq \E^{\overline{\P}}\big[C_p\big(\E^{\overline{\P}}[\|\Mpro_s\|^2|\AlgT_s]\big)^{p/2}\big].
			\end{align} 
			Next we  apply Lemma.~3.1 in Gabih et al.~\cite{Gabih et al (2019) FullInfo}. It yields that there exists a constant $C_M>0$ independent of $t\in[0,T]$ and $\AlgT_s$ such that
			\begin{align}
				\E^{\overline{\P}}[\|\Mpro_s\|^2|\AlgT_s]=\E^{\overline{\P}}[\|\Mpro_s-\drift_s+\drift_s\|^2|\AlgT_s]& \leq 2 \E^{\overline{\P}}[\|\Mpro_s-\drift_s\|^2|\AlgT_s]+2\E^{\overline{\P}}[\|\drift_s\|^2|\AlgT_s]\nonumber\\
				&\leq 2\trace(\E^{\overline{\P}}[\Qpro_s|\AlgT_s])+2\E^{\overline{\P}}[\|\drift_s\|^2|\AlgT_s]\leq C_M.
				\label{estimation_cond_mean}
			\end{align}
			Here we have used  that the conditional covariance matrix $\Qpro$ is bounded according to Lemma \ref{properties_filter}, and  that the drift $\drift$ is Gaussian Ornstein-Uhlenbeck process with bounded moments. Note that the boundedness of $\Qpro$ implies that  $C_M$ can be chosen such that it is independent of the arrival times $T_1,T_2,\ldots$ of the expert opinions, i.e.,~independent of $\AlgT_s$.   Averaging yields $\E^{\overline{\P}}[\|\Mpro_s\|^2]=\E^{\overline{\P}}[\E^{\overline{\P}}[\|\Mpro_s\|^2|\AlgT_s]]\le C_M$ and substituting into inequality \eqref{Moments_Gauss_h} we obtain that all moments of the filter process $\Mpro$ are bounded i.e,  
			$						\E^{\overline{\P}}[\|\Mpro_s\|^p]\leq C_pC_M^{p/2} =: C_{M,p}  \text{ for all } s\in[0,T].$					
			\qed
		\end{proof}

		\section{Proof of Corollary \ref{conv_value}}
		\label{proof_conv_value}
		\begin{proof}
			For $\theta\in(0,1)$ the assertion follows from
			\begin{eqnarray*}
				|{}^\np V(t,y) - V(t,y)| &= & \Big| \sup\limits_{\Pi\in\mathcal{A}} {}^\np \reward(t,y;{}^\np\Pi) - \sup\limits_{\Pi\in\mathcal{A}} \reward(t,y;\Pi) \Big|
				\leq \sup\limits_{\Pi\in\mathcal{A}} | {}^\np \reward(t,y,{}^\np\Pi) - \reward(t,y,\Pi)|
			\end{eqnarray*}
			and  Theorem \ref{conv_reward}.
			Analogously for  $\theta<0$  it holds
			\begin{eqnarray*}
				|{}^\np V(t,y) - V(t,y)| &= & \Big| \inf\limits_{\Pi\in\mathcal{A}} {}^\np \reward(t,y;{}^\np\Pi) - \inf\limits_{\pi\in\mathcal{A}} \reward(t,y;\Pi) \Big| \\
				&= & \Big| \sup\limits_{\Pi\in\mathcal{A}} (-{}^\np \reward(t,y,{}^\np\pi)) - \sup\limits_{\Pi\in\mathcal{A}} (-\reward(t,y,\Pi)) \Big|
				\le \sup\limits_{\Pi\in\mathcal{A}} |{}^\np \reward(t,y;{}^\np\Pi) - \reward(t,y;\Pi)|.\\[-5ex]
			\end{eqnarray*}
			\qed
		\end{proof}
		
		\section{Proof of Corollary \ref{eps_opt}}
		\label{proof_eps_opt}
		\begin{proof}
			It holds
			\begin{align}
				|V(t,y) -\reward(t,y,{}^\np \pi^*)|
				& \le |V(t,y) -{}^\np \reward(t,y,{}^\np \Pi^*)| + |{}^\np \reward(t,y,{}^\np \Pi^*)- \reward(t,y,{}^\np \Pi^*)|  \\
				\label{eps1} 
				&=  |V(t,y) -{}^\np V(t,y)| + |{}^\np \reward(t,y,{}^\np \Pi^*)- \reward(t,y,{}^\np \Pi^*)|,
			\end{align}
			where for the first term on the r.h.s. we used ${}^\np \reward(t,y,{}^\np\Pi^*) = V^\np(t,y)$.
			Using the convergence properties for the reward function given in
			Theorem
			\ref{conv_reward} and for the value function given in Corollary
			\ref{conv_value} we can find for every $\varepsilon>0$  some
			$\np_0\in \mathbb N$, such that for $\np\ge \np_0$ it holds
			\[|V(t,y) -{}^\np V(t,y)| \leq\frac{\varepsilon}{2} \quad\text{and}\quad  |{}^\np \reward(t,y,{}^\np \Pi^*)- \reward(t,y,{}^\np \Pi^*)| \leq\frac{\varepsilon}{2}.\]
			Plugging the above estimates into \eqref{eps1} yields 
			$ |V(t,y) -\reward(t,y,{}^\np \Pi^*)| \le \frac{\varepsilon}{2}
			+ \frac{\varepsilon}{2} = \varepsilon	$ for  $\np\geq \np_0$.
			\qed
		\end{proof}

	\end{appendix}
	\smallskip

	\let\oldbibliography\thebibliography
	\renewcommand{\thebibliography}[1]{%
		\oldbibliography{#1}%
		\setlength{\itemsep}{.15ex plus .05ex}%
	}
	\bibliographystyle{amsplain}

\end{document}